\documentclass[sigconf]{acmart}

\copyrightyear{2017}
\acmYear{2017}
\setcopyright{acmcopyright}
\acmConference{CCS '17}{October 30-November 3, 2017}{Dallas, TX, USA}
\acmPrice{15.00}
\acmDOI{10.1145/3133956.3134030}
\acmISBN{978-1-4503-4946-8/17/10}

\usepackage{setspace,graphicx,epsfig,subfigure,amsfonts,amsmath,amssymb,color,latexsym,multirow,xspace,color,url}
\usepackage[noend]{algorithmic}
\usepackage[linesnumbered,ruled]{algorithm2e}
\usepackage{pifont}
\usepackage{balance}  

\newtheorem{problem}{Problem}
\newcommand{\eat}[1]{}

\newcommand{\p}{\Pi}
\newcommand{\vw}{\text{VIEW}}
\newcommand{\out}{\text{O}}

\newcommand{\dom}{\Sigma}

\newcommand{\negligible}{\text{negl}}
\newcommand{\block}{\mathcal{B}}
\newcommand{\dprl}{DPRL\xspace}
\newcommand{\prl}{PRL\xspace}
\newcommand{\lap}{LP\xspace}

\newcommand{\cmark}{\ding{51}}%
\newcommand{\xmark}{\ding{55}}%

\newcommand{\squishlist}{
   \begin{list}{$\bullet$}
    {
      \setlength{\itemsep}{0pt}
      \setlength{\parsep}{3pt}
      \setlength{\topsep}{3pt}
      \setlength{\partopsep}{0pt}
      \setlength{\leftmargin}{1.5em}
      \setlength{\labelwidth}{1em}
      \setlength{\labelsep}{0.5em} } }
\newcommand{\squishend}{
    \end{list}  }

\begin{document}

\title{Composing Differential Privacy and Secure Computation:\\ A case study on scaling private record linkage}

\author{Xi He}
\affiliation{\institution{Duke University}}
\email{hexi88@cs.duke.edu}
\author{Ashwin Machanavajjhala}
\affiliation{%
  \institution{Duke University}
}
\email{ashwin@cs.duke.edu}
\author{Cheryl Flynn}
\affiliation{%
  \institution{AT\&T Labs-Research}
}
\email{cflynn@research.att.com}
\author{Divesh Srivastava}
\affiliation{%
  \institution{AT\&T Labs-Research}
}
\email{divesh@research.att.com}

\begin{abstract}
Private record linkage (\prl) is the problem of identifying pairs of  records that are similar as per an input matching rule from databases held by two parties that do not trust one another. We identify three key desiderata that a \prl solution must ensure: (1) perfect precision and high recall of matching pairs, (2) a proof of end-to-end privacy, and (3) communication and computational costs that scale subquadratically in the number of input records. We show that all of the existing solutions for \prl -- including secure 2-party computation (S2PC), and their variants that use non-private or differentially private (DP)  blocking to ensure subquadratic cost -- violate at least one of the three desiderata.  In particular, S2PC techniques guarantee end-to-end privacy but  have either low recall or quadratic cost.
In contrast, no end-to-end privacy guarantee has been formalized for
solutions that achieve subquadratic cost.
This is true even for solutions that compose DP and S2PC:
DP does not permit the release of any exact information about the databases, while S2PC algorithms for \prl allow the release of matching records.

In light of this deficiency, we propose a novel privacy model, called {\em output constrained differential privacy}, that shares the strong privacy protection of DP, but allows for the truthful release of the output of a certain function applied to the data. We apply this to \prl, and show that protocols satisfying this privacy model permit the disclosure of the true matching records, but their execution is insensitive to the presence or absence of a single non-matching record. We find that prior work that combine DP and S2PC techniques even fail to satisfy this end-to-end privacy model. Hence, we develop novel protocols that provably achieve this end-to-end privacy guarantee, together with the other two desiderata of \prl. Our empirical evaluation also shows that our protocols obtain high recall, scale near linearly in the size of the input databases and the output set of matching pairs, and have communication and computational costs that are at least 2 orders of magnitude smaller than S2PC baselines.
\end{abstract}

\maketitle

\section{Introduction}

Organizations are increasingly collecting vast amounts of data from individuals to advance science, public health, and resource management and governance. In a number of scenarios, different organizations would like to collaboratively analyze their data in order to mine patterns that they cannot learn from their individual datasets. For instance, hospitals or health workers in neighboring cities might want to identify HIV positive patients who have sought care in multiple cities to quantify the mobility patterns of patients, and hence the spread of the virus. This requires finding patients who occur in multiple databases even though the patient records might not have the same primary key across databases. This problem is called {\em record linkage}, and has been well studied for the last several decades \cite{Christen12:dataMatching, Dong:2013:BDI:2536222.2536253, Getoor:2013:ERB:2487575.2506179}.
In a collaborative analysis across organizations, privacy is always a concern. In particular, one of the collaborating parties, say Hospital A, should not be able to tell whether or not a record is in the database of the other party, say Hospital B, if that record does not appear in the match output. Privacy constraints arise due to concerns from individuals who provide their data, such as hospital patients, or due to contractual or legal obligations that organizations have to the individuals in their data. This has led to a field of research called {\em private record linkage} (\prl).

Traditional \prl techniques aim to solve the linkage problem with a strong privacy goal -- no information should be leaked beyond (a) the sizes of the datasets, and (b) the set of matching records. However, this strong privacy goal (which we call S2PC) \cite{Goldreich:2004:FCV:975541} comes with a high cost. Existing techniques that achieve this goal either require cryptographically secure comparisons of all pairs of records (and hence are inefficient), or are restricted to equi-joins (and thus have very low recall). Hence, we formalize our problem as follows: \emph{given private databases $D_A$ and $D_B$ held by two semi-honest parties, and a matching rule $m$, design a protocol $\p$ that outputs pairs of matching records to both parties and satisfies three desiderata: (1) correctness in terms of perfect precision and high recall of matches, (2) provable end-to-end privacy guarantee, and (3) efficiency in terms of sub-quadratic communication and computational cost in $n$, where $n=\max(|D_A|,|D_B|)$.}
There are two sources of the cost incurred by PRL: (1) the number of cryptographic operations, and (2) the time taken for each cryptographic operation. Our protocols aim to reduce the number of cryptographic operations (i.e., the number of secure pairwise comparisons), the first source of cost,
while using existing techniques to securely compare pairs of records.

Techniques that securely compare all pairs of records (APC) have a quadratic cost and hence fail to meet the efficiency requirement of our problem.
On the other hand, techniques for efficient private set intersection (PSI)
\cite{DBLP:conf/eurocrypt/FreedmanNP04, cryptoeprint:2016:930} satisfy all three desiderata for
equality-like matching rules,
but result in poor recall for general fuzzy matching rules. When records in $D_A$ and $D_B$ come from the same discrete domain, one could expand $D_A$ by adding all records that could potentially match with a record in $D_B$, and then find matches by running PSI on the expanded $D_A$ and $D_B$. However, this technique can be very inefficient:
the expanded databases could be much larger than the input databases for complex matching functions or when data are high dimensional.
A long line of work \cite{Scannapieco:2007:PPS:1247480.1247553, kar15tkde, Inan:2008:HAP:1546682.1547255, Inan:2010:PRM:1739041.1739059, Kuzu:2013:EPR:2452376.2452398,DBLP:conf/icde/CaoRBK15} has considered scaling APC by using {\em blocking}, which is a standard technique for scaling non-private record linkage with a small loss in recall of matching pairs. However, blocking can reveal sensitive properties of input records. We show that such hybrid protocols do not ensure an end-to-end privacy guarantee even in solutions where the blocking step satisfies a strong privacy notion, called differential privacy (DP) \cite{export:64346}.
This negative result is in contrast to
other success stories \cite{DBLP:journals/corr/WaghCM16, 6517175, NIPS2010_0408, Alhadidi:2012:SDF:2359015.2359025, Pettai:2015:CDP:2818000.2818027, Narayan:2012:DDP:2387880.2387895, Goryczka:2013:SMA:2457317.2457343} on composing DP and secure computation.
These settings either consider a client-server model where all data sits on the server or consider aggregated functions across partitioned data where the privacy goals of DP and secure computation do not conflict. In the case of scaling \prl, neither blocking nor DP blocking naturally composes with the strong privacy guarantee of S2PC. To our knowledge, this work presents the first solution to the above open problem, and makes the following contributions:
\squishlist
\item We propose and formalize three desiderata for the \prl problem: (1) correctness, or perfect precision and high recall of matches, (2) provable end-to-end privacy, or insensitivity to the presence or absence of an individual record that is not a matching record,
and (3) efficiency, or communication and computational costs that scale subquadratically in the input size. We show that all of the existing solutions for \prl violate at least one of these three desiderata. (\S~\ref{sec:ps})
\item This motivates us to develop a novel privacy definition, which we call \emph{Output Constrained DP}. Protocols satisfying this notion are allowed to truthfully return the output of a specific function, but must be insensitive to the presence or absence of individual records that do not affect the function output. (\S~\ref{sec:dprl_def})
\item We adapt the notion of Output Constrained DP to the context of \prl. Under this privacy notion, computationally bounded adversaries cannot distinguish two different protocol executions when a single \emph{non-matching} record is replaced by another non-matching record in one of the databases.
 This privacy notion, named \emph{\dprl}, allows protocols to truthfully release the set of matching records. (\S~\ref{sec:dprl2})
\item We show that prior attempts \cite{Inan:2010:PRM:1739041.1739059, Kuzu:2013:EPR:2452376.2452398,DBLP:conf/icde/CaoRBK15} to scale \prl using blocking do not satisfy our privacy definition \dprl (Theorem~\ref{theorem:prldp_limit}), and hence fail to achieve stronger privacy guarantees including differential privacy or S2PC. (\S~\ref{sec:algo})
\item We develop novel protocols for private record linkage that leverage blocking strategies. Our protocols ensure end-to-end privacy (Theorems~\ref{theorem:lap_dprl} and \ref{theorem:gmc}), provide at least as much recall as the non-private blocking strategy (Theorems~\ref{theorem:lap_recall} and \ref{theorem:gmc_eff_recall}), and achieve subquadratic scaling (Theorems~\ref{theorem:efficiency} and \ref{theorem:gmc_eff_recall}).
\item Using experiments on real and synthetic data, we investigate the 3-way trade-off between recall, privacy, and efficiency. Our key findings are: our protocols (1) are at least 2 orders of magnitude more efficient than S2PC baselines, (2) achieve a high recall and end-to-end privacy, and (3) achieve near linear scaling in the size of the input databases and the output set of matching pairs on real and synthetic datasets. (\S~\ref{sec:eval})
\squishend

\section{Problem Setting \& Statement}
\label{sec:ps}

In this section, we formulate our problem: finding pairs of records that are similar as per an input matching rule while ensuring three desiderata: \emph{correctness}, \emph{privacy}, and \emph{efficiency}.  We then discuss prior attempts to solve this problem and how they do not satisfy one or more of the three aforementioned desiderata, thus motivating the need for a novel solution.

\subsection{The Private Record Linkage Problem}
Consider two parties Alice and Bob who have databases $D_A$ and $D_B$. Let records in $D_A$ come from some domain $\dom_A$ and let the records in $D_B$ come from domain $\dom_B$. Let $m: \dom_A \times \dom_B \rightarrow \{0,1\}$ denote a {\em matching rule}, and let $D_A\Join_m D_B$ denote the set of matching pairs $\{(a,b) | a \in D_A, b \in D_b, m(a,b) = 1\}$.
A matching rule can be distance-metric based: two records match if their distance is less than a threshold. For example, Euclidean distance is typically used for numeric attributes, whereas for string attributes, the distance metric is typically based on q-grams  \cite{churches2004blind, churches2004some, schnell2009privacy},
phonetic encoding \cite{DBLP:conf/bci/KarakasidisV09}, or edit distance over strings \cite{Atallah:2003:SPS:1005140.1005147, Ravikumar04asecure, Pang2009}. A matching rule can also be conjunctions of predicates over different types of attributes. For instance, two records match if their names differ by at most 2 characters and their phone numbers differ by at most 1 digit.
Alice and Bob would like to jointly compute $D_A\Join_m D_B$.\footnote{The standard record linkage problem involves learning a matching function in addition to computing the matches. Although the problem considered in this paper and in the private record linkage literature ignores this crucial aspect of record linkage, we have chosen to also use this term for continuity with existing literature on the topic.}

Our goal is to design a protocol $\p$ that Alice and Bob can follow to compute $D_A\Join_m D_B$, while satisfying the following three desiderata -- correctness, privacy and efficiency.
\squishlist
\item {\em Correctness}: Let $O_\p \subseteq D_A \times D_B$ denote the set of pairs output by the protocol $\p$ as the set of matching pairs. The protocol is {\em correct} if
(a) the protocol returns to both Alice and Bob the same output $O_\p$, and
(b) $O_\p = D_A\Join_m D_B$, and incorrect otherwise. Note that
if Alice and Bob indeed receive the same output,
$O_\p$ can only be incorrect in one way -- some matching pairs $(a, b) \in D_A\Join_m D_B$ are not present in $O_\p$.
This ensures perfect precision -- no false positives.
Hence, we quantify the correctness of a protocol $\p$ using a measure called {\em recall}, which is computed as:
\begin{equation} \label{eqn:recall}
r_{\p}(D_A,D_B) = \frac{|O_{\p} \cap (D_A\Join_m D_B)|}{|D_A\Join_m D_B|}.
\end{equation}
We require $\p$ to have a high recall (close to 1). This precludes trivial protocols that output an empty set.
\item {\em Privacy}: We assume that the data in $D_A$ and $D_B$ are sensitive. As part of the protocol $\p$, Alice  would like no one else (including Bob) to learn whether a specific non-matching record $a$ is in or out of $D_A$; and analogously for Bob. This precludes the trivial solution wherein Bob sends $D_B$ to $D_A$ in the clear so that Alice can compute $D_A\Join_m D_B$ using standard techniques in the record linkage literature \cite{Christen12:dataMatching}.
It also precludes the trivial solution wherein Alice and Bob send their records to a trusted third party in the clear who can then compute $D_A\Join_m D_B$.
Formally stating a privacy definition is challenging (as we will see later in the paper) and is a key contribution of this paper. We will assume throughout the paper that Alice and Bob are semi-honest, i.e., they follow the protocol honestly, but are curious about each others' databases.  We also assume that Alice and Bob are computationally bounded, i.e., they are probabilistic polynomially bounded turing machines.
\item {\em Efficiency}: Jointly computing matching records would involve communication and computational cost.
We assume that each record in the database has $O(1)$ length; i.e., it does not grow with $n=\max(|D_A|,|D_B|)$.
The communication and computational costs are bounded below by the output size, i.e. $\Omega(M)$, where $M=|D_A\Join_m D_A|$.
If $M$ is quadratic in $n$, then the costs have to be quadratic in $n$ to ensure high recall. Hence, we consider problems with sub-quadratic output size, and we say that the protocol is efficient if both the communication and computational costs are sub-quadratic in $n$, i.e., $o(n^2)$.
\squishend

We formalize our problem statement as follows.
\begin{problem}[\prl] \label{ps:prl}
Let $D_A$ and $D_B$ be private databases held by two semi-honest parties, and let $m$ be a matching rule. Design a protocol $\p$ that outputs pairs of matching records to both parties such that
(1) $\p$ ensures high recall close to 1,
(2) $\p$ provably guarantees privacy,
and (3) $\p$ has sub-quadratic communication and computational cost.
\end{problem}

\subsection{Prior Work}\label{sec:priorwork}
Before describing our solution, we outline five approaches for the \prl problem from prior work -- APC, PSI, PSI+X, PRL+$\block$ and PRL+$\block_{DP}$. Table~\ref{table:prior} summarizes their (in)ability to satisfy our three desiderata stated in Problem~\ref{ps:prl}. Other related work on composing S2PC and DP is discussed in \S~\ref{sec:relatedwork}.

\begin{table}
\centering
\begin{tabular}{|l|c|c|c|}
\hline
\textbf{Methods}        & \textbf{Correctness} & \textbf{Privacy} & \textbf{Efficiency} \\ \hline
\textbf{APC}          &    \cmark      & \cmark   &  \xmark    \\ \hline
\textbf{PSI}            &  \xmark     & \cmark & \cmark          \\ \hline
  \textbf{PSI+X}        &  \cmark     & \cmark  & \xmark          \\ \hline
\textbf{PRL+$\block$} &   \cmark      &  \xmark     & \cmark   \\ \hline
\textbf{PRL+$\block_{DP}$}       &   \cmark         &   \xmark         &       \cmark    \\ \hline
\end{tabular}
\caption{Summary of Prior Work} \label{table:prior}
\end{table}

\subsubsection{All-Pairwise Comparisons (APC)}
One approach to solve the \prl problem, which we call APC, works as follows: (1) design a secure 2-party algorithm that takes as input a record $a \in D_A$ and a record $b \in D_B$ and outputs to both parties the pair $(a, b)$ if the value of $m(a,b) = 1$ without leaking any additional information, and (2) run the secure comparison algorithm for {\em every} pair of records in $D_A \times D_B$. The secure comparison primitive can be implemented either using garbled circuits \cite{Yao:1986:GES:1382439.1382944} or (partially) homomorphic encryption \cite{Paillier1999}, depending on the matching rule. APC achieves a recall of 1, but requires a quadratic communication and computational cost for $|D_A|\times|D_B|$ secure pairwise comparisons.

APC provides a strong end-to-end privacy guarantee -- it leaks no information other than the sizes of the databases and the set of matching records.
This guarantee is formalized as follows.
\begin{definition}[IND-S2PC \cite{Goldreich:2004:FCV:975541}] \label{def:ind-s2pc}
A 2-party protocol $\p$ that computes function $f$ satisfies IND-S2PC if for any $D_A$,  and for every pair of $D_B$ and $D'_B$ where $f(D_A,D_B) = f(D_A,D'_B)$, the view of Alice during the execution of $\p$ over $(D_A, D_B)$
is computationally indistinguishable from the view over $(D_A,D'_B)$, i.e.  for any probabilistic polynomial adversary $T$,
\begin{eqnarray}
&&Pr[T(\vw^{\p}_A(D_A,D_B))=1] \nonumber \\
&\leq & Pr[T(\vw^{\p}_A(D_A,D'_B))=1] + \negligible(\kappa);
\end{eqnarray}
and the same holds for the view of Bob over $(D_A,D_B)$ and $(D'_A,D_B)$ for $f(D_A,D_B) = f(D'_A, D_B)$.
$\negligible(\kappa)$ refers to any function that is $o(\kappa^{-c})$, for all constants $c$,
and $\vw^{\p}_A(D_A,\cdot)$ ($\vw^{\p}_B(\cdot,D_B)$ resp.) denotes the view of Alice (Bob resp.) during an execution of $\p$.
\end{definition}
The IND-S2PC definition uses $\kappa$ as a ``security'' parameter to control various quantities.
The size of the adversary is polynomial in $\kappa$, and the output of the protocol is at most polynomial in $\kappa$.
The views of the protocol execution are also parameterized by $\kappa$.

In \prl, let $f_{\Join_m}$ be the function that takes as inputs $D_A$ and $D_B$, and outputs a triple $(|D_A|,|D_B|,D_A\Join_m D_B)$.
The view of Alice, $\vw_{A}^{\p}(D_A,\cdot)$, includes ($D_A$, $r$, $m_1$,...,$m_t$), where $r$ represents the outcome of Alice's internal coin tosses, and $m_i$ represents the $i$-th message it has received. The output received by Alice after an execution of $\p$ on ($D_A,D_B$), denoted $\out^{\p}_A(D_A,D_B)$ is implicit in the party's own view of the execution. The view of Bob can be similarly defined. In addition, the output size of $\vw$ will be (at most) polynomial in $\kappa$.
Intuitively, IND-S2PC ensures that the adversary Alice cannot distinguish any two databases $D_B$ and $D'_B$ from her view given the constraint $f(D_A,D_B) = f(D_A,D'_B)$, and the same applies to Bob. This IND-S2PC definition is a necessary condition for the standard simulation-based definition (Theorem~\ref{theorem:sim-ind} in Appendix~\ref{app:sims2pc}).

To summarize, APC  guarantees end-to-end privacy and provides a recall of 1, but violates the efficiency requirement.

\subsubsection{Private Set Intersection (PSI)}
We call the next class of approaches PSI, since they were originally designed for efficient private set intersection. Like APC, PSI also ensures IND-S2PC and the parties only learn the sizes of the databases and the set of matching records. The algorithms are efficient, but only ensure high recall for equality predicate like matching rules \cite{DBLP:conf/eurocrypt/FreedmanNP04,cryptoeprint:2016:930}.

The basic protocol works as follows: Alice defines a polynomial $p(x)$ whose roots are her set of elements $a\in D_A$.
She sends the homomorphic encryptions of the coefficients to Bob.
For each element $b \in D_B$, Bob computes the encrypted values $\tilde{b} = r\cdot p(b)+b$, where $r$ is a random value, and sends them back to Alice.
These values are decrypted by Alice and then matched with $D_A$. If $b\notin D_A$, then the decrypted value of $\tilde{b}$ will be a random value not matching any records in $D_A$; otherwise, it will find a match from $D_A$. The basic protocol described thus far required $O(|D_A| + |D_B|)$ communications and $O(|D_A \times D_B|)$ operations on encrypted values.
\cite{DBLP:conf/eurocrypt/FreedmanNP04} further optimizes the computational cost with Horner's rule and cryptographic hashing to replace a single high-degree polynomial with several low-degree polynomials.  This reduces the computational cost to $O(|D_B|\cdot \ln\ln |D_A|)$, and hence is sub-quadratic in $n$, for $n=\max(|D_A|,|D_B|)$.
State of the art PSI techniques \cite{cryptoeprint:2016:930} further improve efficiency.

PSI techniques are limited to equality like matching functions, and extensions \cite{DBLP:conf/eurocrypt/FreedmanNP04,Ye2010}  allow for matching rules that require exact match on at least $t$ out of $T$ features. However these techniques achieve poor recall for general matching rules. For example, they do not extend to matching rules that involve conjunctions and disjunctions of similarity functions evaluated on multiple attributes. They also do not extend to  complex distance metrics, such as CosineSimilarity(First Name) $> 0.9$ OR CosineSimilarity(Last Name) $> 0.9$,
which are typical in record linkage tasks \cite{Getoor:2013:ERB:2487575.2506179}.

\subsubsection{PSI with Expansion (PSI+X)}
The PSI technique can be used to achieve high recall for general matching rules by using the idea of {\em expansion}.
Suppose $D_A$ and $D_B$ have the same domains, i.e., $\Sigma_A = \Sigma_B = \Sigma$.
For every record $a \in D_A$, one could add all records $a' \in \Sigma$ such that $m(a,a') = 1$ to get an expanded database $D_A^x$. An equi-join between $D_A^x$ and $D_B$ returns the required output $D_A\Join_m D_B$, and satisfies IND-S2PC. However, the expanded dataset can be many orders of magnitude larger than the original dataset making this protocol, PSI+X, inefficient (in the size of the original datasets).
Moreover, enumerating all matches per record is hard for a complex matching function. For instance, if the matching function $m$ can encode Boolean 3-CNF formulas, then finding values for $a$ such that $m(a,a') = 1$ could be an intractable problem.
In such a case, any efficient expansion algorithm may need to enumerate a superset of matches, further increasing the computational cost.
Lastly, even for relatively simple matching functions, we empirically illustrate
low recall of PSI and inefficiency of PSI+X protocols respectively in \S~\ref{sec:eval}.

\subsubsection{\prl with Blocking (PRL+$\block$)} Blocking is commonly used to scale up non-private record linkage. Formally,
\begin{definition}[Blocking ($\block$)] Given $k$ bins $\left\{ \block_0,...,\block_{k-1}\right\}$, records in $D_A$ and $D_B$ are hashed by $\block$ to a subset of the $k$ bins.  The set of records in $D_A$ (respectively $D_B$) falling into the $i^{th}$ bin are represented by $\block_i(D_A)$ (respectively $\block_i(D_B)$). A blocking strategy $\block^S\subseteq [0,k)\times [0,k)$ specifies pairs of bins of $D_A$ and $D_B$ that are compared, i.e. records in $\block_i(D_A)$ are compared with records in $\block_j(D_B)$ if $(i,j)\in \block^S$.
\end{definition}

We sometimes use $\block$ to refer to the entire blocking algorithm
as well as the blocking functions used in the algorithm.
We refer to the set of pairs of records that are compared by a blocking strategy as {\em candidate matches}. A blocking strategy $\block^S$ is \emph{sub-quadratic} if the number of candidate matches
$$cost_{\block^S}(D_A,D_B)=\sum_{(i,j)\in \block^S} |\block_i(D_A)| |\block_j(D_B)|$$ is $o(n^2)$, for $n=\max(|D_A|,|D_B|)$.
Blocking techniques are useful as a pre-processing step  \cite{Scannapieco:2007:PPS:1247480.1247553, kar15tkde, Inan:2008:HAP:1546682.1547255}
to achieve sub-quadratic efficiency and high recall.
We can use blocking as a pre-processing step for APC -- secure comparison is performed only for the candidate matches -- resulting in an efficient protocol with high recall.
However, the blocking strategy itself can leak information about the presence or absence of a record in the database. This was illustrated using an attack by Cao et al.~\cite{DBLP:conf/icde/CaoRBK15}.
This is because the number of candidate matches can vary significantly even if $D_B$ and $D'_B$ differ in only one record.
We formally prove this negative result for a large class of blocking techniques which use locality sensitive hashing (LSH). A majority of the hash functions used by blocking algorithms like q-gram based hash signatures \cite{Al-Lawati:2005:BPR:1077501.1077513} or SparseMap \cite{Scannapieco:2007:PPS:1247480.1247553} are instances of LSH.

\begin{definition}[Locality Sensitive Hashing (LSH)\cite{Gionis:1999:SSH:645925.671516}]
A family of functions $H$ is said to be $(d_1,d_2, p_1,p_2)$-sensitive, where $d_2>d_1$ and $p_1>p_2$, if for all $h\in H$, (1)  if $dist(a,b) \leq d_1$, then $Pr[h(a)=h(b)]\geq p_1$, and (2) if $dist(a,b) > d_2$, then $Pr[h(a)=h(b)] \leq p_2$.
\end{definition}

An LSH-based blocking considers a set of bins where each bin consists of records with the same hash values for all $h\in H$. A popular blocking strategy is to compare all the corresponding bins, and results in a set of candidate matches $\{(a,b)| h(a)=h(b) \forall h\in H, a\in D_A, b\in D_B \}$. In general,  we can show that any LSH based blocking cannot satisfy IND-S2PC.
\begin{theorem}\label{theorem:lsh}
An LSH based blocking with a family of $(d_1,d_2,p_1,p_2)$-sensitive hashing functions $H$ cannot satisfy IND-S2PC.
\end{theorem}
The proof can be found in Appendix~\ref{app:prlblocking}.

\subsubsection{\prl with DP Blocking (PRL+$\block_{DP}$)}
Differential privacy has arisen as a gold standard for privacy in situations where it is ok to reveal statistical properties of datasets but not reveal properties of individuals. An algorithm satisfies differential privacy if its output does not significantly change when adding/removing or changing a single record in its input. More formally,
\begin{definition}[($\epsilon,\delta$)-Differential Privacy\cite{export:64346}]\label{def:dp}
A randomized mechanism $M: \mathcal{D}\rightarrow \mathcal{O}$ satisfies $(\epsilon,\delta)$-differential privacy (DP) if
\begin{equation}
Pr[M(D) \in O] \leq e^{\epsilon} Pr[M(D')\in O]+\delta
\end{equation}
for any set $O\subseteq \mathcal{O}$
and any pair of neighboring databases $D,D'\in \mathcal{D}$
such that $D$ and $D'$ differ by adding/removing a record.
\end{definition}

A recent line of work has designed differentially private blocking algorithms as a preprocessing step to APC. DP hides the presence or absence of a single record, and hence the number of candidate matches stays roughly the same on $D_B$ and $D'_B$ that differ in a single record. While this approach seems like it should satisfy all three of our desiderata, we have found that none of the protocols presented in prior work (on DP Blocking) \cite{Inan:2010:PRM:1739041.1739059, Kuzu:2013:EPR:2452376.2452398,DBLP:conf/icde/CaoRBK15} provide an end-to-end privacy guarantee. In fact, each paper in this line of work finds privacy breaches in the prior work.
We also show in the proof of Theorem~\ref{theorem:prldp_limit} (Appendix~\ref{app:prldp}) that even the most recent of these protocols in \cite{DBLP:conf/icde/CaoRBK15} does not satisfy an end-to-end privacy guarantee.
This is because of a fundamental disconnect between the privacy guarantees in the two steps of these algorithms.
DP does not allow learning any fact about the input datasets with certainty, while IND-S2PC (and \prl protocols that satisfy this definition) can reveal the output of the function $f$ truthfully. On the other hand, while DP can reveal aggregate properties of the input datasets with low error, protocols that satisfy IND-S2PC are not allowed to leak any information beyond the output of $f$. Hence, DP and IND-S2PC do not naturally compose.

To summarize, none of the prior approaches that attempt to solve Problem~\ref{ps:prl} satisfy all three of our desiderata.  Approaches that satisfy a strong privacy guarantee (IND-S2PC) are either inefficient or have poor recall. Efficient \prl with blocking or DP blocking fail to provide true end-to-end privacy guarantees.
A correct conceptualization of an end-to-end privacy guarantee is critical for achieving correctness, privacy and efficiency.
Hence, in the following sections, we first define an end-to-end privacy guarantee for \prl to address this challenge (\S~\ref{sec:dprl}),
and then present algorithms in this privacy framework to achieve sub-quadratic efficiency and high recall (\S~\ref{sec:algo}).

\section{Output Constrained DP}
\label{sec:dprl}

Designing efficient and correct algorithms for \prl is challenging and non-trivial because there is no existing formal privacy framework that enables the trade-off between
 correctness, privacy, and efficiency.  In this section, we propose a novel privacy model to achieve this goal.

\subsection{Output Constrained Differential Privacy}\label{sec:dprl_def}
Both IND-S2PC (Def.~\ref{def:ind-s2pc}) and DP (Def.~\ref{def:dp}) ensure the privacy goal of not revealing information about individual records in the dataset. However, there is a fundamental incompatibility between the two definitions. IND-S2PC reveals the output of a function truthfully; whereas, nothing truthful can be revealed under differential privacy. On the other hand, DP reveals noisy yet accurate (to within an approximation factor) aggregate statistics about all the records in the dataset; but, nothing other than the output of a pre-specified function can be revealed under IND-S2PC.

The difference between these privacy definitions can be illustrated by rephrasing the privacy notions in terms of a distance metric imposed on the space of databases. Without loss of generality, assume Alice is the adversary. Let $\mathcal{G}=(V,E)$ denote a graph, where $V$ is the set of all possible databases that Bob could have and $E$ is a set of edges that connect neighboring databases. The distance between any pair of databases is the shortest path distance in $\mathcal{G}$. Intuitively, the adversary Alice's ability to distinguish protocol executions on a pair of databases $D_B$ and $D'_B$ is larger if the shortest path between the databases is larger.

DP can be represented by the set of edges that connect neighboring databases that differ in the presence or absence of one record, $|D_B\backslash D_B'\cup D'_B\backslash D_B|=1$. This means, any pair of databases $D_B$ and $D'_B$ are connected in this graph by a path of {\em finite} length that is equal to the size of their symmetric difference. While an adversary can distinguish protocol executions between some pair of ``far away'' databases, the adversary can never tell with certainty whether the input was a specific database. On the other hand, under IND-S2PC, {\em every pair} of databases that result in the same output for $f(D_A,\cdot)$ for a given $D_A$ are neighbors. However, there is neither an edge nor a path between databases that result in different outputs. Thus the output constraint divides the set of databases into disjoint complete subgraphs  (in fact equivalence classes).

\begin{example}\label{eg:neighbor}
Consider databases with domain $\{1,2,3,4,5,6\}$. Given $D_A=\{1,2\}$, the graph $\mathcal{G}$ for the database instances for $D_B$ are shown in Figure~\ref{fig:graphexamples}. For the graph of differential privacy in Figure~\ref{fig:example_dp}, every pair of database instances that differ in one record is connected by an edge and form a neighboring pair. For instance, $D_B=\{1\}$ and $D'_B=\{1,2\}$ are neighbors under DP. Figure~\ref{fig:example_s2pc} considers an output which consists of the size of $D_B$ and the intersection between $D_B$ and $D_A$. Hence, all the instances in $\mathcal{G}_{IND-S2PC}$ have the same datasize and have the same intersection with $D_A=\{1,2\}$. For example, the fully connected 6 database instances all have 2 records, but have no intersection with $D_A$. The instance $\{1,2\}$ has no neighboring databases, as it is same as the output, and hence none of the records in this database instance requires privacy protection.

Comparing these two graphs, we can see that all instances in $\mathcal{G}_{DP}$ are connected, and hence an adversary can not distinguish protocol executions on any pair of databases with certainty, but is allowed to learn statistical properties (with some error). This is not true under $\mathcal{G}_{IND-S2PC}$, where some instances are disconnected.
For instance, an adversary can distinguish between protocol executions on $\{1,2\}$ and $\{1,5\}$ since they give different outputs when matched with $D_A$.
\end{example}

From Example~\ref{eg:neighbor}, it is clear that the privacy guarantees given by DP and IND-S2PC are different. To ensure scalable record linkage with formal privacy guarantees, we need the best of both worlds: \emph{the ability to reveal records that appear in the match truthfully, the ability to reveal statistics about non-matching records, and yet not reveal the presence or absence of individual non-matching records in the dataset}. Hence, we propose a weaker, but end-to-end, privacy definition for the two party setting.

\begin{figure}[t]
\centering
\subfigure[$\mathcal{G}_{DP}$]{
\includegraphics[scale=0.3]{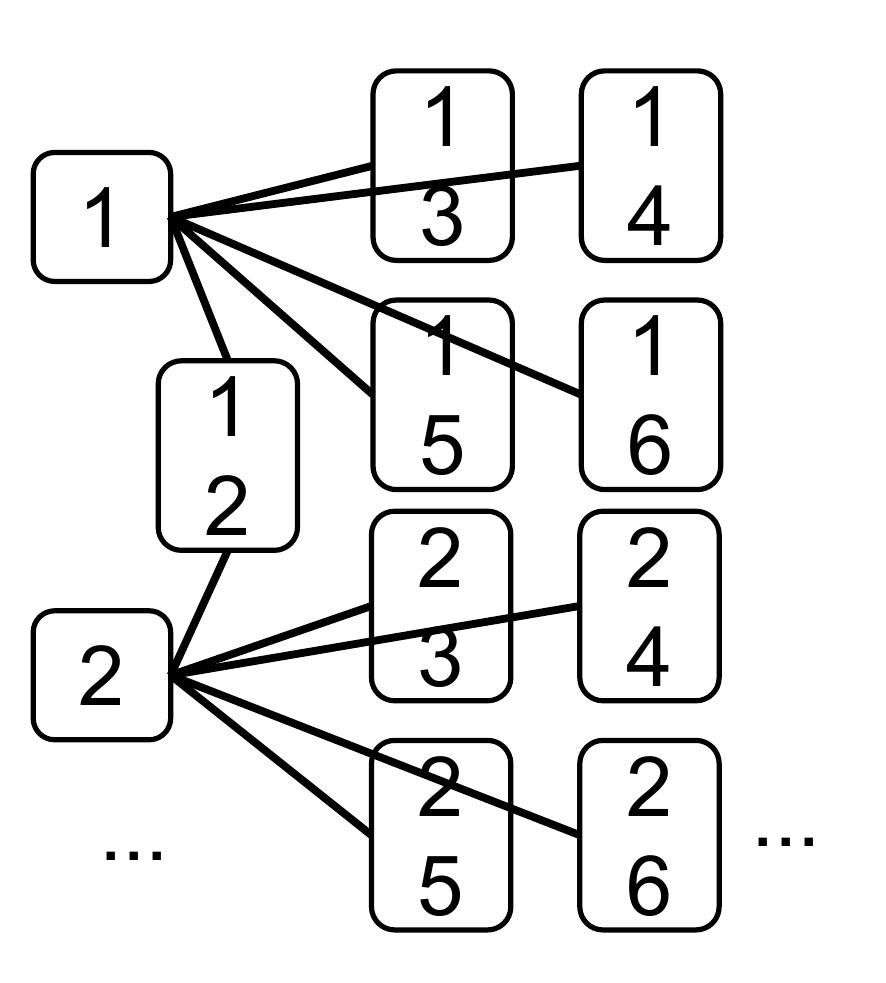} \label{fig:example_dp}
}
~
\subfigure[$\mathcal{G}_{IND-S2PC}$]{
\includegraphics[scale=0.3]{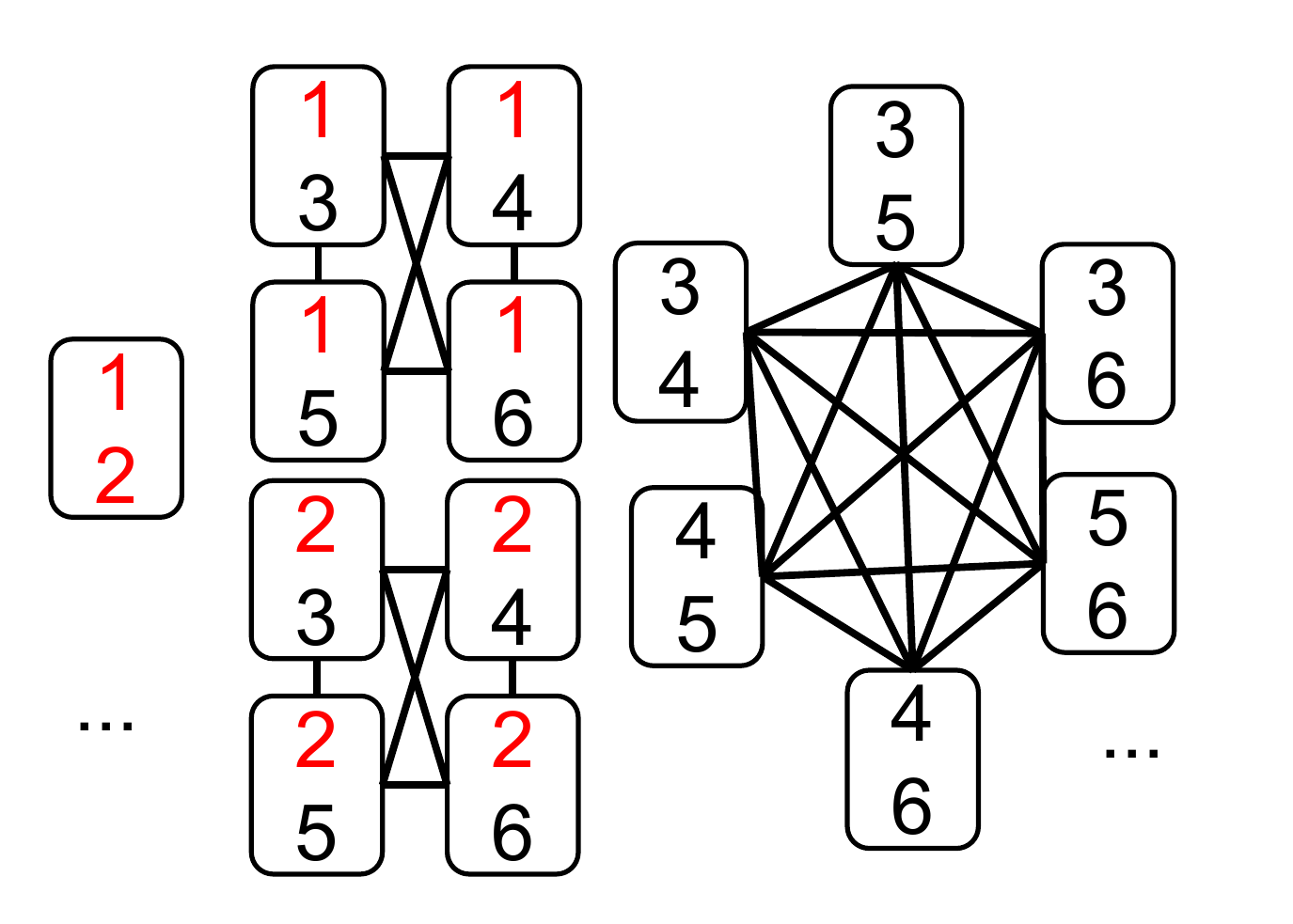}\label{fig:example_s2pc}
}
\caption{Neighboring databases for (a) DP, and (b) IND-S2PC for Example~\ref{eg:neighbor}.}\label{fig:graphexamples}
\end{figure}

\begin{definition}[$f$-Neighbors] \label{def:cneighbor} Given function
$f:\mathcal{D}_A \times \mathcal{D}_B \rightarrow \mathcal{O}$
and $D_A\in \mathcal{D}_A$.
For any pairs of datasets $D_B,D'_B$,
let $\triangle(D_B, D'_B) = D_B\backslash D'_B \cup D'_B \backslash D_B$.
This is the symmetric difference between $D_B$ and $D'_B$, and is the set of records that must be deleted and added to $D_B$ to get $D'_B$.
$D_B$ and $D'_B$ are neighbors w.r.t to $f(D_A, \cdot)$, denoted by $\mathcal{N}(f(D_A, \cdot))$ if
\squishlist
\item[(1)] $f(D_A,D_B) = f(D_A,D'_B)$,
\item[(2)] $\triangle(D_B,D'_B) \neq \emptyset$, and
\item[(3)] there is no database $D_B''\in \mathcal{D}_B$, where $f(D_A,D_B) = f(D_A,D''_B)$, such that
$\triangle(D_B,D''_B) \subset \triangle(D_B,D'_B)$.
\squishend
$\mathcal{N}(f(\cdot,D_B))$ is similarly defined.
\end{definition}
The third condition ensures that $D_B$ and $D'_B$ are minimally different
in terms of record changes.

\begin{definition} [Output Constrained DP]\label{def:output}
A 2-party \prl protocol $\p$ for computing function $f:\mathcal{D}_A \times \mathcal{D}_B \rightarrow \mathcal{O}$ is $(\epsilon_A,\epsilon_B, \delta_A,\delta_B,f)$-constrained differential privacy (DP) if for any $(D_B,D'_B)\in \mathcal{N}(f(D_A,\cdot))$, the views of Alice during the execution of $\p$ to any probabilistic polynomial-time adversary $T$ satisfies
\begin{eqnarray}
&&Pr[T(\vw^{\p}_A(D_A,D_B))=1] \nonumber \\
&\leq & e^{\epsilon_B} Pr[T(\vw^{\p}_A(D_A,D'_B))=1] + \delta_B
\end{eqnarray}
and the same holds for the views of Bob with $\epsilon_A$ and $\delta_A$.
\end{definition}

If $\epsilon_A=\epsilon_B=\epsilon, \delta_a=\delta_b=\delta$, we simply denote it as $(\epsilon,\delta,f)$-constrained DP. Similar to DP, Output Constrained DP satisfies composition properties that are useful for protocol design.

\begin{theorem}[Sequential Composition]\label{theorem:seq}
Given $\p_1$ is $(\epsilon_1,\delta_1,f)$-constrained DP, and $\p_2$ is $(\epsilon_2,\delta_2,f)$-constrained DP, then applying these two protocols sequentially, i.e. $\p_2(D_A, D_B,\p_1(D_A,D_B))$ satisfies $(\epsilon_1+\epsilon_2, \delta_1+\delta_2,f)$-constrained DP.
\end{theorem}

\begin{theorem}[Post-processing] \label{theorem:pp}
Given $\p$ is $(\epsilon,\delta,f)$-constrained DP, and let $\out^{\p}(D_A,D_B)$ be the output after the execution of $\p$, then any probabilistic polynomial (in $\kappa$) function  $g(\out^{\p}(D_A,D_B))$ satisfies $(\epsilon,\delta,f)$-constrained DP.
\end{theorem}

See Appendix~\ref{app:seqproof} and ~\ref{app:ppproof} for the proofs of Theorem~\ref{theorem:seq} and Theorem~\ref{theorem:pp} respectively.
Output constrained DP inherits
other desirable properties from DP,
for instance, its robustness to attacks \cite{psc_smith:2015, He:2014:BPT:2588555.2588581}.
We omit details due to space constraints.

\subsection{Differential Privacy for Record Linkage}\label{sec:dprl2}
\prl can be a direct application of Output Constrained Differential Privacy by considering $f_{\Join_m}$. We have the following theorem to define the neighboring databases for \prl.
\begin{theorem}[Neighbors for \prl]\label{theorem:neighbor_prl}
Given the function $f_{\Join_m}$ in \prl, if $(D_B,D'_B)\in \mathcal{N}(f_{\Join_m}(D_A, \cdot))$ for a given $D_A\in\mathcal{D}$, then $|D_B|=|D'_B|$, $D_B$ and $D'_B$ must differ in only one pair of non-matching records with respect to the given $D_A$, i.e. $D'_B=D_B-b+b'$ and $b\neq b'$, where $m(b,a)=0$ and $m(b',a)=0$ for all $a \in D_A$.
\end{theorem}
\begin{proof}
The output constraint $f_{\Join_m}(D_A,D_B)=f_{\Join_m}(D_A,D'_B)$ implies that $|D_B|=|D'_B|$ and $D_A\Join_m D_B = D_A \Join_m D'_B$.
If $D_B$ and $D'_B$ differ in a matching record, then their matching outputs with a given $D_A$ are different. Hence $D_B$ and $D'_B$ must
differ in one or more non-matching records.
In addition, to ensure $|D_B|=|D'_B|$, the number of non-matching records added to $D_B$ to get $D'_B$ must be the same as the number of non-matching records removed from $D_B$. If $\triangle(D_B,D'_B)$ contains more than one pair of record additions and deletions, a subset of $\triangle(D_B,D'_B)$ can give a valid $D''_B$ such that $f_{\Join_m}(D_A,D_B)=f_{\Join_m}(D_A,D''_B)$. Hence, a neighboring pair $D_B,D'_B$ differ by exactly one pair of non-matching records.
\end{proof}

Next we define the privacy guarantee that allows us to design efficient \prl protocols with provable privacy guarantees.
\begin{definition}[\dprl] \label{def:dprl}
A 2-party \prl protocol $\p$ for computing function $f_{\Join_m}:\mathcal{D}_A \times \mathcal{D}_B \rightarrow \mathcal{O}$ is $(\epsilon_A,\epsilon_B, \delta_A,\delta_B)$-\dprl if
$\p$ satisfies $(\epsilon_A,\epsilon_B,\delta_A,\delta_B,f_{\Join_m}$)-constrained DP.
\end{definition}

\subsection{Related Privacy Definitions}
In this section we discuss related privacy definitions and their connections with \dprl.
First, both \dprl and IND-S2PC assume a computationally bounded model. We show that \dprl is a weaker guarantee than IND-S2PC.
\begin{theorem}\label{theorem:dprl_ind-s2pc}
All IND-S2PC protocols for record linkage satisfy (0,$\negligible(\kappa)$)-\dprl.
\end{theorem}
\begin{proof}
IND-S2PC for record linkage is equivalent to \dprl with $\epsilon=0$ and $\delta=\negligible(\kappa)$. The $\delta$ in \dprl is always greater than $\negligible(\kappa)$ but smaller than $o(1/n)$.
\end{proof}
Hence, APC, PSI, and PSI+X techniques that satisfy IND-S2PC, guarantee (0,$\negligible(\kappa)$)-\dprl as well.

Indistinguishable computationally differential privacy (IND-CDP-2PC) \cite{Mironov:2009:CDP:1615970.1615981} is another privacy notion under a computationally bounded model, and is a direct extension of DP to the two party setting where both parties are computationally bounded. \dprl is weaker than IND-CDP-2PC. Formally
\begin{theorem}\label{theorem:dprl_cdp}
If a protocol for record linkage satisfies $\epsilon/2$-IND-CDP-2PC, then it satisfies $(\epsilon,\delta)$-\dprl.
\end{theorem}
The factor 2 arises since neighboring databases protected by \dprl have a symmetric difference of 2,
while neighboring databases under IND-CDP-2PC have a symmetric difference of 1.
The detailed proof can be found in Appendix~\ref{app:cdp}.

Blowfish Privacy \cite{He:2014:BPT:2588555.2588581} generalizes differential privacy to problems where constraints on the input database must hold (e.g., when certain query answers have been released by the database exactly). Output Constrained DP, including \dprl, is an extension of Blowfish in two ways:
(1) from a computationally unbounded model to a computationally bounded model;
(2) from a single-party setting to a two-party setting.
Note that with the output constraint $f_{\Join_m}(D_A,D_B)=f_{\Join_m}(D_A,D'_B)$ for record linkage, the number of different records between neighboring databases $D_B$ and $D'_B$ is only two.
This is not necessarily true for other applications of Output Constrained DP, or Blowfish Privacy. This property is desirable for DP based algorithms since larger distances between neighboring databases typically require larger perturbation to hide the difference between neighbors resulting in poorer utility.

Another instantiation of Blowfish privacy, called Protected DP \cite{Kearns26012016}, aims to ensure the privacy of
a protected subpopulation. In contrast, an unprotected ``targeted'' subpopulation receives no privacy guarantees.
In \dprl, one could think of the non-matching records as the protected subpopulation, and the matching records as the targeted subpopulation. However, unlike in Protected DP, in DPRL the set of protected records are learned as an output of the \dprl protocol, and hence are not available as an input to the protocol like the targeted subpopulation in the Protected DP algorithms.

\section{Protocols for \dprl }
\label{sec:algo}

In this section, we introduce protocols that satisfy \dprl and permit a 3-way trade-off between correctness, privacy and efficiency. We first present a class of protocols
that achieves $(\epsilon,\negligible(\kappa))$-\dprl by using a blocking strategy that satisfies {\em local} differential privacy (DP).
Though these protocols achieve high recall with a sufficiently small privacy parameter, they only achieve a constant factor speedup in efficiency. Next, we present the Laplace Protocol (\lap) that achieves all three desiderata of high recall, privacy and subquadratic efficiency.
This protocol hides non-matching records by adding Laplace noise to the blocking strategy. We also show that
attempts from prior work to use Laplace noise in blocking fail to
satisfy \dprl (Theorem~\ref{theorem:prldp_limit}). Moreover, we design a Sort \& Prune (SP) heuristic that is used in conjunction with \lap (as well as the local DP based protocols) and helps additionally tradeoff efficiency and recall. Finally, we present the Greedy Match \& Clean heuristic optimization (GMC), that can further improve efficiency. All the protocols presented in this section are proven to satisfy \dprl.

\subsection{Local DP Protocol}
Let $\block$ be a blocking that randomly hashes records into a pre-specified set of $k$ bins, such that for all $i \in [1\ldots k]$,
\begin{equation}
\Pr[\block(b) = i] \leq e^{\epsilon} \Pr[\block(b')=i].
\end{equation}
Such a blocking $\block$ satisfies $\epsilon$-local DP (as defined in Appendix~\ref{app:localdpdef}), since each record is perturbed \emph{locally} independent of the other records. We show that protocols that combine a local differentially private blocking with IND-S2PC protocols for record linkage can achieve $(\epsilon, \negligible(\kappa))$-\dprl.

\begin{theorem}\label{theorem:localdp}
All IND-S2PC protocols for record linkage with $\epsilon$-local differentially private blocking satisfies $(\epsilon,\negligible(\kappa))$-\dprl.
\end{theorem}
The proof can be found in Appendix~\ref{app:localdp}. Such local differentially private protocols can be constructed from well known local differentially private algorithms based on randomized response (RR) \cite{Dwork:2014:AFD:2693052.2693053} or the Johnson-Lindenstrauss (JL) transformation \cite{Blocki:2012:JTI:2417500.2417899}, where each record is hashed independent of others. We refer the reader to Appendix~\ref{app:localdp} for a concrete blocking algorithm based on RR. We show that while this algorithm permits high recall and privacy, it does not improve efficiency by more than a constant factor (a function of $\epsilon$) (Theorem~\ref{theorem:localdp_constantfactor}).  Whether any local DP based blocking algorithms can achieve subquadratic efficiency is an interesting open question.

\subsection{Laplace Protocol (\lap)}\label{sec:basicLap}
\subsubsection{Algorithm Description}
In this protocol, Alice and Bob agree on a blocking function $\block$ with $k$ bins and strategy $\block^S$, which we take as input to the protocol.
The Laplace Protocol (\lap, as shown in Algorithm~\ref{algo:lap}) works by inserting a carefully chosen number of \emph{dummy} records into each bin of the blocking strategy such that the bin sizes are differentially private. While candidate matches may contain dummy records, they do not contribute to the output set of matches, because the dummy records do not match {\em any} record. These candidate matches are then securely matched using an IND-S2PC algorithm.

\begin{algorithm}[t]
{\small
    \SetKwInOut{Input}{Input}
    \SetKwInOut{Output}{Output}
    \Input{$D_A$,$D_B$, $\epsilon_A,\epsilon_B$, $\delta_A,\delta_B$, $\block$(including $\block^S$)}
    \Output{$O$}
    \ \tcp{Alice performs the following: }
    	\ $\tilde{\block}(D_A) \leftarrow LapNoise( D_A, \block, \epsilon_A, \delta_A)$ \;
  \ \tcp{Bob performs the following: }
	\ $\tilde{\block}(D_B) \leftarrow LapNoise( D_B, \block, \epsilon_B, \delta_B)$ \;
	\ \tcp{Alice and Bob perform the following:}
    \ $O=\emptyset$  \;
    	\ \tcp{Sort \& prune $\block^S$  (\S~\ref{sec:prune})}
	\For{$(i,j)\in \block^S$}{
			\For{$a\in \tilde{\block}_i(D_A)$ and $b\in \tilde{\block}_j(D_B)$}{
				\ Add $SMC(a,b)$ to $O$ \;
    			}
    			\ \tcp{Greedy match \& clean (\S~\ref{sec:greedy})}
    }

    \ return $O$\;
    \caption{Laplace Protocol (\lap)}  \label{algo:lap}
}
\end{algorithm}

\begin{algorithm}[t]
    \underline{function LapNoise} $(D, \block,\epsilon,\delta)$\;
    	\For{$\block_{i}\in \block$}{
		$\eta_i \sim Lap(\epsilon,\delta,\Delta\block)$ \;
    		$\tilde{\block}_{i}(D) \leftarrow$  add $\eta_i^+=\max (\eta_i,0)$ dummy records to $\block_{i}(D)$\;
	}
     \ return $\tilde{B}(D)$\;
    \caption{Add Laplace Noise}  \label{algo:lapNoise}
\end{algorithm}

In the first step (Lines~1-4) of the protocol shown in Algorithm~\ref{algo:lap}, Alice and Bob take their inputs $D_A$ and $D_B$, the agreed blocking protocol $\block$, and privacy parameters $\epsilon_A$, $\epsilon_B$, $\delta_A$, and $\delta_B$ as input, and compute \emph{noisy} bins $\tilde{\block}(D_A)$ and $\tilde{\block}(D_B)$ respectively. The noisy bins are constructed as follows (Algorithm~\ref{algo:lapNoise}).
Records in $D$ are first hashed into bins according to the blocking protocol $\block$, and $\block(D)$ denotes the set of bins of records from $D$.
Then the counts of the bins are perturbed using noise drawn from a truncated and discretized Laplace distribution, such that the noisy counts satisfy $(\epsilon, \delta)$-DPRL.
The Laplace noise depends on not only the privacy parameters $\epsilon$ and $\delta$, but also the sensitivity of the given blocking protocol $\block$.

\begin{definition}[Sensitivity of $\block$]\label{def:deltablock}
The sensitivity of the blocking strategy $\block$ for Bob, denoted by $\Delta\block_B$ is
\begin{equation}
\max_{D_A\in\mathcal{D}} \max_{(D_B,D'_B)\in \mathcal{N}(f_{\Join_m}(D_A,\cdot))} \sum_{i=0}^k ||\block_i(D_B)|-|\block_i(D'_B)||, \nonumber
\end{equation}
the maximum bin count difference between $D_B$ and $D'_B$ for any $(D_B,D'_B)\in \mathcal{N}(f_{\Join_m}(D_A,\cdot))$ for all $D_A\in \mathcal{D}$. $\Delta\block_A$ for Alice is similarly defined.
\end{definition}
If the hashing of $\block$ is the same for Alice and Bob, then $\Delta\block_A = \Delta\block_B =\Delta\block$. We assume this in our paper. If $\block$ hashes each record to at most $k'$ bins, then $\Delta\block = 2k'$.

\begin{definition}[$Lap(\epsilon,\delta, \Delta\block)$] \label{def:lapbias}
A random variable follows the $Lap(\epsilon,\delta,\Delta\block)$ distribution if it has a probability density function
\begin{equation}\label{eqn:pdf}
\Pr[\eta = x] = p \cdot e^{-(\epsilon/\Delta\block) |x-\eta^0|},~ \forall x\in \mathbb{Z},
\end{equation}
where $p=\frac{e^{\epsilon/\Delta\block}-1}{e^{\epsilon/\Delta\block}+1}$, and
$\eta^0 = -\frac{\Delta\block\ln ((e^{\epsilon/\Delta\block}+1)(1-(1-\delta)^{1/\Delta\block}))}{\epsilon}$.
\end{definition}

This distribution has a mean of $\eta_0$ and takes both positive and negative values. \lap draws a noise value $\eta$ from this distribution, and truncates it to 0 if $\eta$ is negative. Then, $\eta$ dummy records are added to the bin.
These dummy records lie in an expanded domain, such that they do not match with any records in the true domain.

After Alice and Bob perturb their binned records, they will initiate secure matching steps to compare candidate matches, i.e. records in $\tilde{\block}_i(D_A)\times \tilde{\block}_j(D_B)$ if $(i,j)\in \block^{S}$.
For each candidate match $(a,b)$, Alice and Bob participate in a two party secure matching protocol $SMC(a,b)$ that outputs the pair $(a,b)$ to both Alice and Bob if $m(a,b) = 1$ (true matching pair) and null otherwise. Secure matching can be implemented either using garbled circuits \cite{Yao:1986:GES:1382439.1382944} or (partially) homomorphic encryption \cite{Paillier1999}, depending on the matching rule (see Appendix~\ref{app:eg_smc} for an example).

\subsubsection{Correctness Analysis}
Compared to the original non-private blocking protocol $\block$, no records are deleted, and dummy records do not match any real record. Hence,
\begin{theorem}\label{theorem:lap_recall}
Algorithm~\ref{algo:lap} gives the same recall as the non-private blocking protocol $\block$ it takes as input.
\end{theorem}

\subsubsection{Privacy Analysis}
Next, we show that \lap satisfies \dprl.
\begin{theorem}\label{theorem:lap_dprl}
Algorithm~\ref{algo:lap} satisfies $(\epsilon_A,\epsilon_B,\delta_A,\delta_B)$-\dprl.
\end{theorem}
\begin{proof}
We prove privacy for Bob (the proof for Alice is analogous). In this protocol, Alice with input data $D_A$ has a view consisting of (1) the number of candidate matching pairs arising in each $(i,j) \in \block^S$, (2) the output for each candidate matching pair.
Algorithm~\ref{algo:lap} is the composition of two steps: (a) add dummy records to bins, and (b) secure comparison of records within bins.

 Consider a neighboring pair $(D_B,D'_B) \in N(f_{\Join_m}(D_A,\cdot))$ for a given $D_A$. By Theorem~\ref{theorem:neighbor_prl}, $D_B$ and $D'_B$ differ in only one non-matching record with respect to $D_A$, i.e. $D'_B=D_B-b_*+b_*'$ and $b_*\neq b_*'$, where $m(b_*,a)=0$ and $m(b_*',a)=0$ for all $a \in D_A$. $D_B$ and $D'_B$ can differ by at most $\Delta\block$ in their bin counts. We show in Lemma~\ref{lemma:ratio} (Appendix)
that  Algorithm~\ref{algo:lapNoise} adds a sufficient number of dummy records to hide this difference:  with probability $1-\delta_B$, the probabilities of generating the same noisy bin counts for Bob, and hence the same number of candidate matching pairs consisting in each $(i,j)\in\block^S$ from $D_B$ and $D'_B$  are bounded by $e^{\epsilon_B}$. Thus, Step~(a) ensures $(\epsilon_B,\delta_B)$-\dprl for Bob.
Given a fixed view from Step~(a) which consists of the noisy bin counts and encrypted records from $\tilde{\block}(D_B)$, Alice's view regarding the output for each candidate matching pair $(a,b)$ is the same. The encrypted records for a given noisy bin counts can only differ in $b_*$ and $b'_*$, but both of them lead to the same output for each candidate matching, because they do not match any records in $D_A$. Each secure pairwise comparison satisfies ($0,\negligible(\kappa)$)-\dprl, and since there are at most $n^2$ comparisons (recall $\kappa >n= \max(|D_A|,|D_B|)$). Thus Step~(b) satisfies ($0,\negligible(\kappa)$)-\dprl.

Therefore, using similar arguments for Alice and sequential composition, we get that  Algorithm~\ref{algo:lap} satisfies \dprl.
\end{proof}

\begin{theorem}\label{theorem:neglap}
If Algorithm~\ref{algo:lap} (\lap) takes $\eta_0= \ln^2n \cdot \Delta\block/\epsilon$ for Eqn.~\eqref{eqn:pdf}, then \lap satisfies $(\epsilon_A,\epsilon_B,o(1/n^{k}),o(1/n^k))$-\dprl, for any $k>0$,  where $n=\max(|D_A|,|D_B|)$.
\end{theorem}
\begin{proof} (sketch)
Taking $\eta_0= \ln^2n \cdot \Delta\block/\epsilon$, the failing probability $\delta =1-(1-\frac{1}{n^{\ln n} (e^{\epsilon/\Delta\block}+1)})^{\Delta\block}\leq\frac{c}{n^{\ln n}}$ for some constant $c$ (in terms of $\epsilon,\Delta\block$). Hence $\delta = o(1/n^k)$ for all $k>0$.
\end{proof}

\lap only adds non-negative noise to the bin counts.
One could instead add noise that could take positive and negative values, and suppress records if the noise is negative.
We call this protocol \lap-2. This is indeed the protocol proposed by prior work \cite{Inan:2010:PRM:1739041.1739059, Kuzu:2013:EPR:2452376.2452398,DBLP:conf/icde/CaoRBK15} that combined APC with DP blocking. However, we show that this minor change in \lap results in the protocol violating \dprl (even though the noise addition seems to satisfy DP)! Hence, \lap-2 also does not satisfy IND- CDP-2PC (by Theorem~\ref{theorem:dprl_cdp}).

\begin{theorem}\label{theorem:prldp_limit}
For every non-negative $\epsilon,\delta<\frac{p^{\Delta\block}}{2e^{\epsilon}}$, there exists a pair of neighboring databases for which \lap-2 does not ensure $(\epsilon,\delta)$-\dprl, where $p=\frac{e^{\epsilon/\Delta\block}-1}{e^{\epsilon/\Delta\block}+1}$.
\end{theorem}
\begin{proof}(sketch)
The output of the record suppression step is dependent on the ratio between the matching and non-matching records in the bin. This introduces a correlation between the matching and non-matching records. Consider a neighboring pair $D_B$ and $D'_B$ that differ by a non-matching pair $(b_*,b'_*)$ for a given $D_A$. If $b_*$ is in a bin full of non-matching records with $D_A$, and $b'_*$ is in a bin full of matching records with $D_A$ (except $b'_*$). $D_B$ is more likely to output all matching pairs than $D'_B$ if some record is suppressed. The detailed proof can be  found in Appendix~\ref{app:prldp}.
\end{proof}

\subsubsection{Efficiency Analysis}
Last, we present our result on the efficiency of \lap.
Note that the communication and computational costs for \lap are
the same as $O(cost_{\block^S})$, where $cost_{\block^S}$ is the number of candidate matches,
if you consider the communication and computational costs associated with a single secure comparison as a constant.
Hence, we analyze efficiency in terms of the number of candidate pairs $cost_{\block^S}$ in \lap.

\begin{theorem}\label{theorem:efficiency}
Given a blocking protocol $\block$ with $k$ bins and blocking strategy $\block^S$, such that the number of candidate matches for $D_A$ and $D_B$, $cost_{\block^S}(D_A,D_B)$, is sub-quadratic in $n$, i.e. $o(n^2)$, where  $n=\max(|D_A|,|D_B|)$. If (1) the number of bins $k$ is $o(n^c)$ for $c<2$, and (2) each bin of a party is compared with $O(1)$ number of bins from the opposite party, then the expected number of candidate matches in Algorithm~\ref{algo:lap} is sub-quadratic in $n$.
\end{theorem}

\begin{proof}
Given $\epsilon$ and $\delta$, the expected number of dummy records added per bin $\mathbb{E}(\eta^+)$ is a constant denoted by $c_{\eta}$ (Def.~\ref{def:lapbias}). Each bin of a party is compared with at most $c_b$  bins from the opposite party, where $c_b$ is a constant.
The number of candidate matches in \lap is a random variable, denoted by $COST$, with expected value
\begin{eqnarray}
\mathbb{E}(COST) &=& \sum_{(i,j)\in \block^S} \mathbb{E}(|\tilde{\block}_i(D_A)| |\tilde{\block}_j(D_B)|) \nonumber \\
&=& \sum_{(i,j)\in \block^S} |\block_i(D_A)| |\block_j(D_B)|  + \sum_{(i,j)\in \block^S}   \mathbb{E}(\eta_i^+) \mathbb{E}(\eta_j^+) \nonumber\\
&& + \sum_{(i,j)\in \block^S} (\mathbb{E}(\eta_i^+) |\block_j(D_B)| + \mathbb{E}(\eta_j^+) |\block_i(D_A)|)  \nonumber \\
&<& cost_{\block^S}(D_A,D_B) + c_{\eta}^2 c_b k + 2 c_{\eta}c_b n.  \nonumber
\end{eqnarray}
Since $cost_{\block^S}(D_A,D_B)$ and $k$ are sub-quadratic in $n$, $\mathbb{E}(COST)$ is also sub-quadratic in $n$. When $\delta$ is a negligible term as defined in Theorem~\ref{theorem:neglap}, the noise per bin is $O(\ln^2 n)$. As $k$ is $o(n^c)$ for $c<2$, the expected value of $COST$ is still sub-quadratic in $n$.
\end{proof}
Conditions (1) and (2) in the above theorem are satisfied by, for instance, sorted neighborhood, and distance based blocking \cite{Christen12:dataMatching} (we use the latter in our experiments).
While the asymptotic complexity of \lap is sub-quadratic, it performs at least a constant number of secure comparisons for each pair $(i,j) \in \block^S$ even if there are no real records in $\block_i(D_A)$ and $\block_j(D_B)$.
We can reduce this computational overhead with a slight loss in recall (with no loss in privacy) using a  heuristic we describe in the next section.

\subsection{Sort \& Prune $\block^S$ (SP)} \label{sec:prune}
Algorithm~\ref{algo:lap} draws noise from the same distribution for each bin, and hence the expected number of dummy records is the same for every bin. The bins with higher noisy counts will then have a higher ratio of true to dummy records. This motivates us to match candidate pairs in bins with high noisy counts first. Instead of comparing bin pairs in $\block^S$ in a random  or index order, we would like to sort them based on the noisy counts of $\tilde{\block}(D_A)$ and $\tilde{\block}(D_B)$. Given a list of descending thresholds $\bar{t}=[t_1,t_2,t_3\ldots]$, the pairs of bins from the matching strategy $\block^S$ can be sorted into groups denoted by $\block^{S,t_l}$ for $l=1,2,\ldots$, where $$\block^{S,t_l} = \{|\tilde{\block}_i(D_A)| >t_ l \wedge |\tilde{\block}_j(D_B)|>t_ l  | (i,j)\in \block^S   \}.$$ Each group consists of bin pairs from $\block^S$ with both noisy counts greater than the threshold.

We let the thresholds $\bar{t}$ be the deciles of the sorted noisy bin sizes  of $\tilde{\block}(D_A)$ and $\tilde{\block}(D_B)$.
As the threshold decreases, the likelihood of matching true records instead of dummy records drops for bins. Alice and Bob can stop this matching process before reaching the smallest threshold in $\bar{t}$.
If the protocol stops at a larger threshold, the recall is smaller.
In the evaluation, if the protocol stops at 10\% percentile of the noisy bin counts, the recall can reach more than 0.95.
This allows a trade-off between recall and efficiency for a given privacy guarantee.
We show that this step also ensures  \dprl.

\begin{corollary}\label{col:prune}
Algorithm~\ref{algo:lap} with sort \& prune step (SP) satisfies $(\epsilon_A,\epsilon_B,\delta_A,\delta_B)$-\dprl.
\end{corollary}
\begin{proof}
Similar to the proof in Theorem~\ref{theorem:lap_dprl}, Alice with input data $D_A$ has a view consisting of (1) the number of candidate matching pairs arising in each $(i,j) \in \block^S$, and (2) the output for each candidate matching pair. As SP is a post-processing step based on the noisy bin counts, which is part of Alice's original view,  the overall protocol still satisfies the same \dprl guarantee by Theorem~\ref{theorem:pp} (post-processing).
\end{proof}

We next present an optimization that also uses a form of post-processing to  significantly reduce the number of secure pairwise comparisons in practice, but whose privacy analysis is more involved than that of SP.

\subsection{Greedy Match \& Clean (GMC)} \label{sec:greedy}
\lap executes a sequence of secure comparison protocols, one per candidate pair. After every comparison (or a block of comparisons), Alice and Bob learn a subset of the matches $O$. Based on the current output $O$, Alice and Bob can greedily search matching pairs in the clear from their respective databases (Lines~5,10 in Algorithm~\ref{algo:greedy}), and add the new matching pairs to the output set $O$ until no new matching pairs can be found. In addition, Alice and Bob can remove records in the output from the bins $\tilde{\block}(D_A)$ and $\tilde{\block}(D_B)$ to further reduce the number of secure pairwise comparisons (Lines~4,9). We can see that this optimization step is not simply post-processing, because it makes use of the true record in plain text for matching. In traditional differential privacy, when the true data is used for computation, the privacy guarantee decays. However, we show that this is not true for the GMC step in the setting of \dprl.

\begin{algorithm}[t]
{\small
    \SetKwInOut{Input}{Input}
    \Input{$O$, $\tilde{\block}(D_A)$, $\tilde{\block}(D_B)$}
    	\Repeat{$O$ received by Alice has no updates}{
	\ \tcp{Alice performs the following: }
        \ $O_A\leftarrow \pi_{A} O$, $O_B\leftarrow \pi_{B} O$ \;
	\ $\tilde{\block}(D_A) \leftarrow \tilde{\block}(D_A) - O_A$ \;
	\ $O' \leftarrow PlainMatch(O_B, \tilde{\block}(D_A))$ \;
	\ Add $O'$ to $O$ and send $O$ to Bob \;
	\ \tcp{Bob performs the following: }
        \ $O_A\leftarrow \pi_{A} O$, $O_B\leftarrow \pi_{B} O$ \;
	\ $\tilde{\block}(D_B) \leftarrow \tilde{\block}(D_B) - O_B$ \;
	\ $O' \leftarrow PlainMatch(O_A, \tilde{\block}(D_B))$ \;
	\ Add $O'$ to $O$ and send $O$ to Alice \;
	}
    \caption{Greedy match and clean}  \label{algo:greedy}
}
\end{algorithm}

\begin{theorem}\label{theorem:gmc}
Algorithm~\ref{algo:lap} with the greedy match \& clean step (GMC) in Algorithm~\ref{algo:greedy} satisfies  $(\epsilon_A,\epsilon_B,\delta_A,\delta_B)$-\dprl.
\end{theorem}

\begin{proof}
First consider the privacy for Bob. Alice  with input data $D_A$, has a view consisting of (1) the number of candidate matching pairs arising in each $(i,j) \in \block^S$, (2) the output for each candidate matching pair, (3) the output from plaintext comparisons with output records.

Consider a neighboring pair $(D_B,D'_B) \in N(f_{\Join_m}(D_A,\cdot))$ for a given $D_A$. By Theorem~\ref{theorem:neighbor_prl}, $D_B$ and $D'_B$ differ in only one non-matching record with respect to $D_A$, i.e. $D'_B=D_B-b_*+b_*'$ and $b_*\neq b_*'$, where $m(b_*,a)=0$ and $m(b_*',a)=0$ for all $a \in D_A$. $D_B$ and $D'_B$ can differ by at most $\Delta\block$ in their bin counts. Similar to the proof for Theorem~\ref{theorem:lap_dprl}, the first step of the protocol adds dummy records to bins, and satisfies $(\epsilon_B, \delta_B)$-\dprl.

In the second step, given a fixed view $\vw^*$ from the first step which consists of the noisy bin counts and encrypted records from $\tilde{\block}(D_B)$, Alice's view regarding the output for each candidate matching pair $(a,b)$ is the same regardless $(a,b)$ are compared securely or in plaintext. Alice's view regarding the output from plaintext comparisons with the records in the output set is also the same for a fixed $\vw^*$ from the first step.  The encrypted records for a given noisy bin counts can only differ in $b_*$ and $b'_*$, and they will never be pruned away. Both of them also lead to the same output for secure pairwise comparisons or plaintext comparisons, because they do not match any records in $D_A$. Thus Step~(b) satisfies $(0,\negligible(\kappa)$-\dprl.

Therefore, using similar arguments for Alice and sequential composition, we get that  Algorithm~\ref{algo:lap} satisfies \dprl.
\end{proof}

With the same privacy guarantee,
\lap with the GMC step can even improve the efficiency of \lap without sacrificing recall.
\begin{theorem}\label{theorem:gmc_eff_recall}
\lap with the greedy match \& clean step (GMC) performs no more secure pairwise comparisons than \lap, and outputs at least as many matching pairs as \lap.
\end{theorem}
We refer the reader to Appendix~\ref{app:gmc_eff_recall} for the proof.
Both SP and GMC are also applicable on the local DP based protocols for the similar reasoning.
Hence, we will only show how each optimization helps improve the efficiency of the basic \lap in the evaluation.

\section{Evaluation}\label{sec:eval}
We empirically evaluate the correctness, privacy, and efficiency of the protocols proposed in \S~\ref{sec:algo}. Our experiments demonstrate the following results:
\squishlist
\item The Laplace Protocol (\lap, which includes all the optimizations) proposed in \S~\ref{sec:algo} is over 2 orders of magnitude more efficient than the baseline approaches while still achieving a high recall and end-to-end privacy. (\S~\ref{sec:eff_baseline})
\item At any given level of privacy, \lap incurs a computational cost that is near linear in the input database size. (\S~\ref{sec:eff_baseline})
\item Greedy match \& clean and Sort \& prune optimization help reduce communication and computation costs. The former results in 50\% lower cost than unoptimized LP in some cases. (\S~\ref{sec:eval_opt})
\item We explore the 3-way trade-offs between correctness, privacy, and efficiency of \lap. (\S~\ref{sec:eval_tradeoff})
\squishend

\subsection{Evaluation Setup}\label{sec:eval_setup}
\subsubsection{Datasets and Matching Rules}

{\em Taxi dataset} {\bf (Taxi):} To simulate linkage in the location domain, we extract location distribution information from the TLC Trip Record Data \cite{tlctaxi}. Each record includes a pickup location in latitude-longitude coordinates (truncated to 6 decimal places) and the date and hour of the pickup time.
Taking the original dataset as $D_A$, we create  $D_B$ by perturbing the latitude-longitude coordinates of each record in $D_A$ with random values uniformly drawn from $[-\theta,+\theta]^2$, where $\theta=0.001$. Each day has approximately 300,000 pickups.
The data size can be scaled up by increasing the number of days, $T$. We experiment with $T=1,2,4,8,16$, with $T=1$ being the default.
Any pair of records $a,b\in \dom$ are called a match if they have the same day and hour, and their Euclidean distance in location is no larger than $\theta$.
The location domain is within the bounding box (40.711720N, 73.929670W) and (40.786770N, 74.006600W).  We project the locations into a uniform grid of $16 \times 16$ cells with size  $0.005\times 0.005$.
A blocking strategy $\block^S$ based on the pickup time and grid is applied to both datasets, resulting in $(16\times 16\times 24T)$ bins. $\block^S$ compares pairs of bins that are associated with the same hour, and corresponding/neighboring grid cells. Thus, each bin in $\block(D_A)$ is compared with 9 bins in $\block(D_B)$.

{\em Abt and Buy product dataset} {\bf (AB):} These datasets are synthesized from the online retailers Abt.com and Buy.com \cite{abtbuy} who would like to collaboratively study the common products they sell as a function of time.
Each record in either dataset consists of a product name, brand and the day the product was sold. The product names are tokenized into trigrams, and hashed into a bit vector with a bloom filter having domain $\dom = \left\{0,1\right\}^{50}$. We consider 16 brands, and sample 5,000 records per day from the original datasets for Abt and Buy each. The data size can be scaled up with $T$ for $T=1,2,4,8,16$, with 1  being the default for $T$. Any pair of records $a,b\in \dom$ are called a match if (a) they are sold on the same day, (b) they are of the same brand, and (c) the hamming distance between their vectorized names is no more than $\theta=5$. A blocking strategy hashes records having the same value for day and brand into the same bin, resulting in $16T$ bins, and compares records falling in the corresponding bins.

\subsubsection{Protocols:}  We evaluate four \dprl protocols: (1) Laplace protocol (\lap), (2) all-pairwise comparisons (APC),  (3) private set intersection (PSI), and (4) PSI with expansion (PSI+X). The default \lap consists of the basic protocol described in Algorithm~\ref{algo:lap} along with optimization steps (SP and GMC) in \S~\ref{sec:prune} and \ref{sec:greedy}.

\subsubsection{Metrics:}
There are three dimensions in the trade-off space: correctness, privacy and efficiency. The correctness of a protocol is measured by the {\em recall}, which is the fraction of the matching pairs output by the algorithm, as defined in Eqn.~\eqref{eqn:recall}, with larger values close to 1 being better. The {\em privacy} metric is specified in advance for each algorithm using parameters $\epsilon,\delta$. For AP, PSI, and PSI+X, $\epsilon=0$ and $\delta=\negligible(\kappa)$ by Theorem~\ref{theorem:dprl_ind-s2pc}.
We consider $\epsilon_A=\epsilon_B=\epsilon$ and $\delta_A=\delta_B$ for $\epsilon\in \{0.1,0.4,1.6\}$ and $\delta\in \{10^{-9},10^{-7},10^{-5}\}$ for \lap.
The default value for $\epsilon$ and $\delta$ is $1.6$ and $10^{-5}$, respectively. Finally, we define {\em efficiency} of APC and \lap protocols for a given dataset as the number of secure pairwise comparisons, and denote this by \emph{cost}. The cost of PSI and PSI+X can be estimated as $\gamma n\ln\ln(n)$, where $\gamma$ is the expansion factor, or the ratio of sizes of the expanded and true databases. This represents the number of operations on encrypted values. For PSI, $\gamma$ is 1. We use the number of secure comparison/operations on encrypted values rather than the wallclock times as a measure of efficiency, since these operations dominate the total time. We discuss wallclock times in more detail in \S~\ref{sec:clock}.

\begin{figure}[t]
\centering
\includegraphics[scale=0.57]{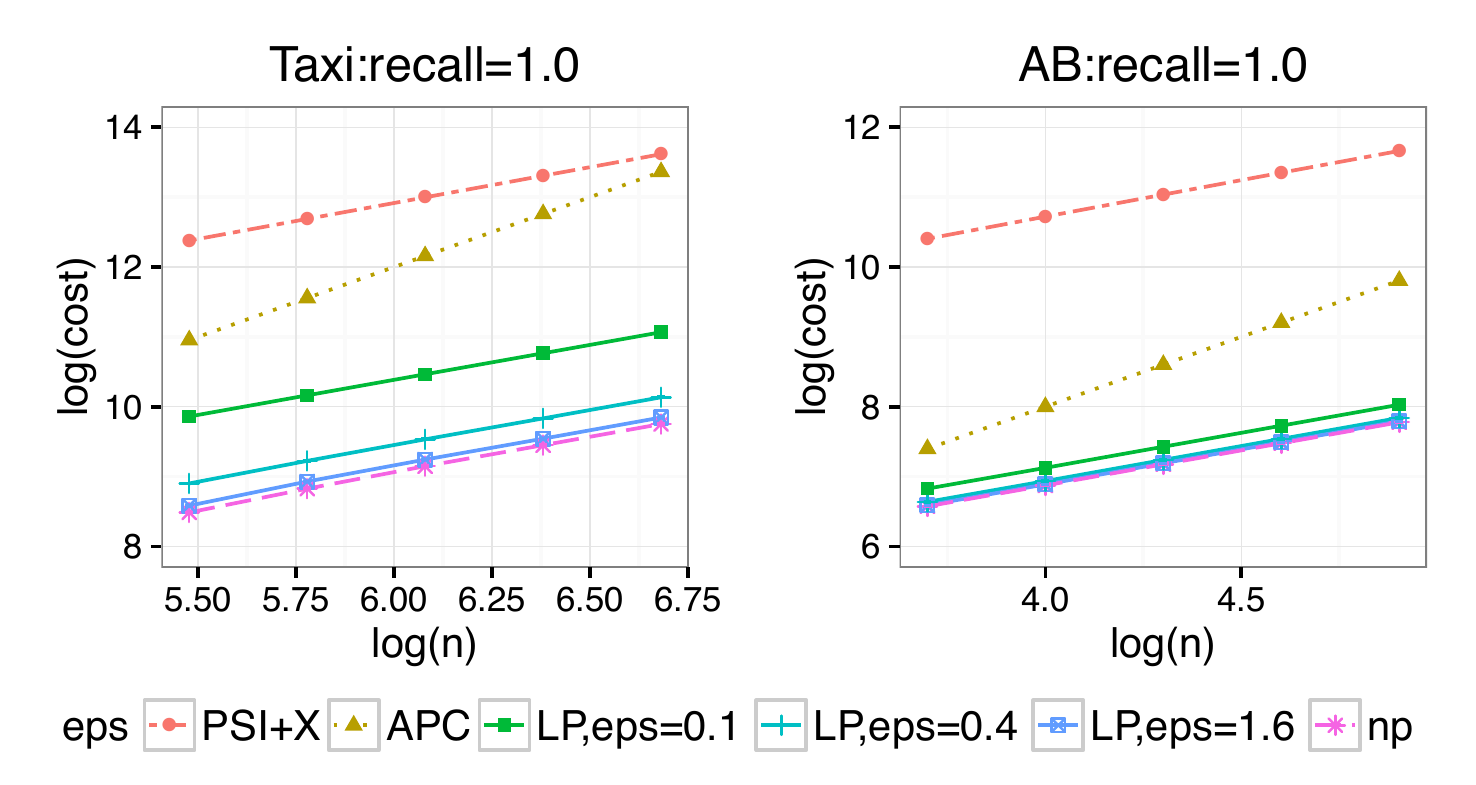}
\caption{The average $\log\mbox{(cost)}$ of \lap, APC, PSI+X and non-private matching (np) for the Taxi and AB datasets vs $\log\mbox{(data size)}$.
\lap give lower costs than the baselines PSI+X and APC for all values of $\epsilon=0.1,0.4,1.6$ and $\delta=10^{-5}$, and scales near linearly.
}\label{fig:greedy_recall100}
\end{figure}

\begin{figure*}[th!]
  \centering
\includegraphics[scale=0.58]{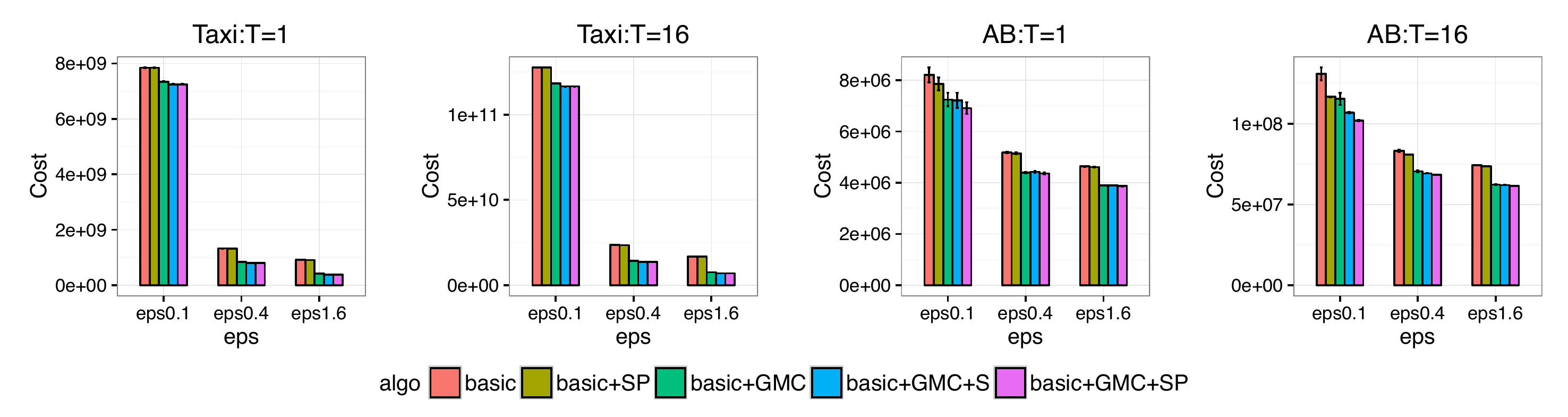}
\caption{The average cost with standard deviation of \lap protocols with five settings: (1) `basic' -- the basic \lap protocol in Algorithm~\ref{algo:lap}, (2) `basic+SP' -- the basic \lap with sort \& prune step in \S~\ref{sec:prune}, (3) `basic+GMC' -- the basic protocol with greedy match \& clean step in Algorithm~\ref{algo:greedy}, (4) `basic+GMC+S' -- the previous protocol with addition sorting step described in \S~\ref{sec:prune}, and (5) `basic+GMC+SP' -- the protocol stops at recall less than full recall.
}\label{fig:opt}
\end{figure*}

\subsection{Results and Discussions}\label{sec:findings}

\subsubsection{Efficiency and scalability} \label{sec:eff_baseline}
In this section, we empirically investigate how \lap scales as the data size increases ($T\in\{1,2,4,8,16\}$) in comparison to baselines APC and PSI+X, when all the algorithms achieve 100\% recall. We do not include PSI as its recall is close to 10\%. \lap is evaluated at privacy parameter $\epsilon \in \{1.6, 0.4, 0.1\}$ and fixed $\delta = 10^{-5}$.  At each $\epsilon$, we report the average number of candidate pairs for LP over 10 runs for each value of $T$.
To achieve 100\% recall, PSI+X expands each record $b$ in $D_B$ to every other record $b'$ within a $\theta$-ball around $b$.  We add 2,369,936 records per record in the AB dataset, and $1000^2\pi$ records per record in the Taxi dataset.

In Figure~\ref{fig:greedy_recall100}, we report the log(base 10) value of the average cost, $\log(cost)$, with respect to the log value of data size $\log(n)$ for PSI+X, APC, and \lap with varying $\epsilon$ and the non-private setting (np) when they achieve a recall of 1.0.
Results for Taxi are shown on the left, and AB are shown on the right.
For both datasets, the baseline methods, PSI+X and APC, have data points and line segments above \lap for the plotted data size range.
When the Taxi dataset has a size of $10^{5.5}$, \lap at $\epsilon=0.1$ costs an order of magnitude less than APC, as shown by the leftmost brown point (APC) and blue point (\lap,eps=0.1) in Figure~\ref{fig:greedy_recall100}(left). As the data size increases, the gap between APC and \lap gets larger. When data size increases by 16 times (the right most points in the plots), \lap at $\epsilon=0.1$ costs over 2 orders of magnitude less than APC.  When $\epsilon$ increases, the cost of \lap shifts downward towards the non-private setting (np). When $\epsilon=1.6$, \lap has 3 orders of magnitude lower cost than APC for the given range of data sizes. The line for np is the lower bound for \lap, where no dummy records are added to the bins. Similar observations are found in Figure~\ref{fig:greedy_recall100}(right) for the AB dataset, where \lap improves APC by up to 2 orders for the plotted data size range.

PSI+X has a much larger cost than both APC and \lap, mainly due to the fact that the expansion factor is far larger than the data size. We also observe that the lines that pass through the points of APC for both Taxi and AB datasets have a slope of 2, which corresponds to the quadratic communication and computational cost of APC.
LP and PSI+X have slopes of values slightly larger than 1, and thus are linear time. Thus, for sufficiently large data sizes, PSI+X can beat APC. However, we do not expect PSI+X to beat LP due to the large expansion factor. Similar results are observed when the protocol stops before achieving full recall (Figure~\ref{fig:greedy_recall95} in Appendix~\ref{app:eval}).

\subsubsection{Optimization steps}\label{sec:eval_opt}
We next study the effectiveness of the optimization steps for \lap. We study 5 protocols as shown below:
\squishlist
\item `basic': the basic \lap Algorithm~\ref{algo:lap} with no heuristic optimizations;
\item `basic+SP': the basic \lap with the sort \& prune step (SP). SP stops the protocol when the threshold reaches the 10\% percentile of the noisy bin counts of $\tilde{\block}(D_A)$ and $\tilde{\block}(D_B)$. Together with the sorting step, bins pairs with insufficient counts can be pruned away, resulting in a recall slightly smaller than the highest possible recall;
\item `basic+GMC': the greedy match \& clean step (GMC) in Algorithm~\ref{algo:greedy} is applied to the basic \lap;
\item `basic+GMC+S': in addition to the previous protocol, bins are sorted in order of size. Pruning is omitted so that the highest possible recall is achieved;
\item  `basic+GMC+SP': the same protocol as `basic+GMC+S', except it prunes the bins with counts in the bottom 10\% percentile.
\squishend

Hence, the default \lap can be also denoted by `basic+GMC+S' if recall is 1.0 and `basic+GMC+SP' if recall is less than 1.0.

In Figure~\ref{fig:opt}, we report the average cost with the standard deviation across 10 runs of the above mentioned protocols at $\epsilon=0.1, 0.4, 1.6$ and $\delta=10^{-5}$ for the Taxi and AB datasets when $T=1$ and $T=16$. Several interesting observations arise from this plot.

First, the most significant drop in cost is due to GMC. The protocols with the greedy step have smaller cost than other protocols for all $\epsilon$ and datasets. For the Taxi datasets at $T=1$ or $T=16$, `basic+GMC' saves the cost of `basic' by over 50\% when $\epsilon=1.6$. As $\epsilon$ decreases, these relative savings reduce because more dummy records are added and cannot be matched or removed by this greedy step. For the AB datasets, `basic+GMC' reduces the cost of `basic' by up to 16\% at $\epsilon=1.6$ and 11\% at $\epsilon=0.1$.

Next, adding the sorting step to GMC (GMC +S) improves upon GMC when the data sizes are large (T=16). For instance, when $\epsilon=0.1$ and $T=16$, `basic+GMC+S' can further bring the cost down by approximately $8.0\times 10^6$ candidate pairs for the AB datasets, and by $2.0\times10^9$ for the Taxi datasets.

Third, the cost of `basic+GMC+SP' is reported at a recall reaching above 0.95. The reduction with respect to `basic+GMC+S' is relatively small, but the absolute reduction in cost is significant in some setting. For instance, the number of candidate pairs is reduced by $5.0\times 10^6$ for the AB datasets when $\epsilon=0.1$ and $T=16$.

Last, for the AB dataset at $T=16$, `basic+SP' has a smaller variance in cost than `basic' at $\epsilon=0.1$. Similarly, `basic+GMC+SP' has a smaller variance in cost than `basic+GMC'. This implies the sort \& prune step can help prune away bins, and hence reduce the variance introduced by dummy records.

\begin{figure*}[t]
\centering
\subfigure[Vary $\epsilon\in \{1.6,0.4,0.1\}$, $\delta=10^{-5}$]{
\includegraphics[scale=0.55]{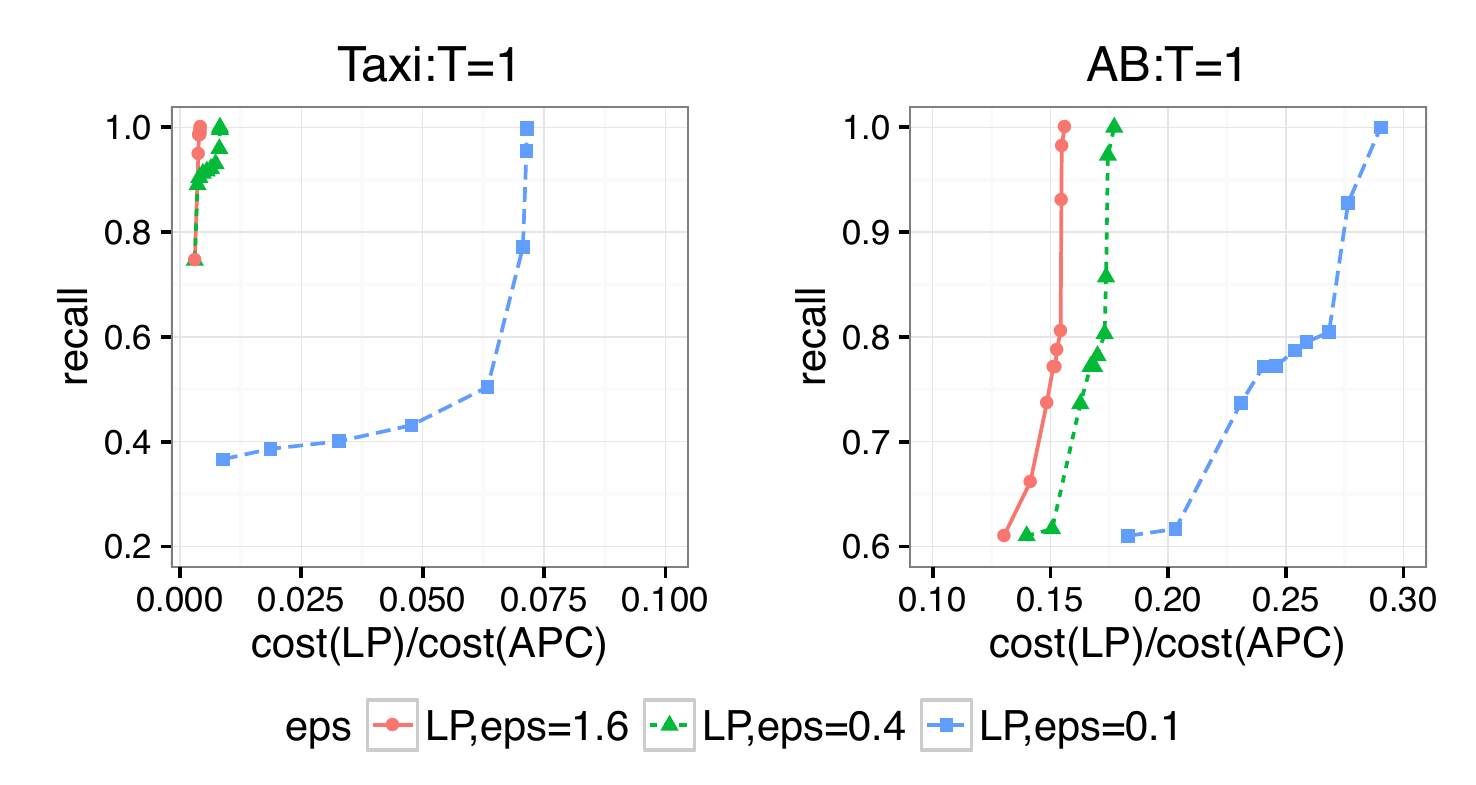}\label{fig:vary_eps}
}
~
\subfigure[Vary $\delta\in\{10^{-5}, 10^{-7}, 10^{-9}\}$, $\epsilon=1.6$]{
\includegraphics[scale=0.55]{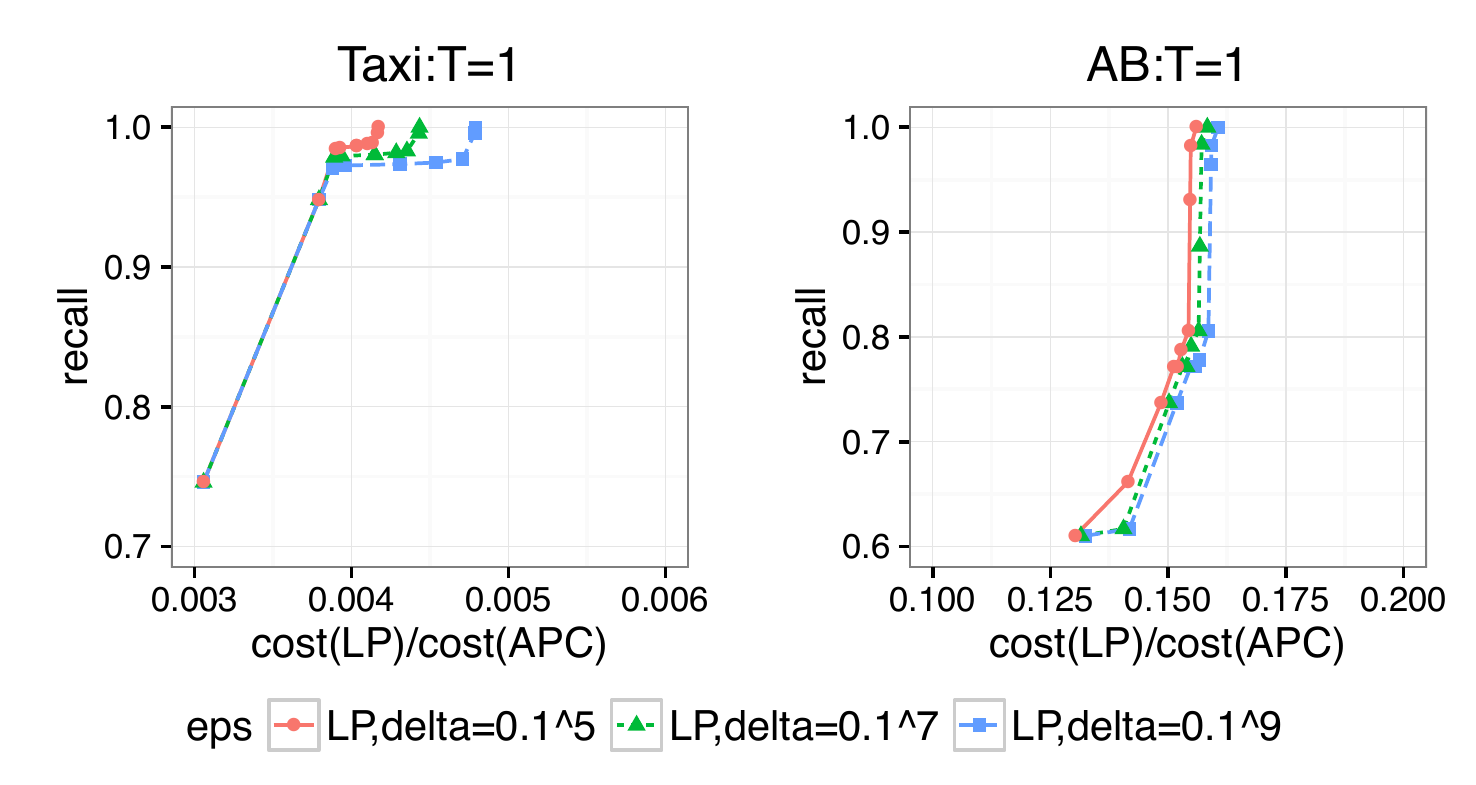}\label{fig:vary_delta}
}
\caption{\lap with varying privacy settings plotted over the default Taxi datasets and AB datasets. Each trade-off line between recall and the efficiency (cost(\lap)/cost(APC)) corresponds to the default \lap at a privacy setting $(\epsilon,\delta)$. Figure~\ref{fig:vary_eps} varies $\epsilon$ and Figure~\ref{fig:vary_delta} varies $\delta$.
}\label{fig:varyprivacy}
\end{figure*}

\subsubsection{Three-way trade-offs}\label{sec:eval_tradeoff}
All the \dprl baseline methods including APC, PSI and PSI+X, have a fixed and strong privacy guarantee where $\epsilon=0$ and $\delta=\negligible(\kappa)$. Hence, each baseline has a single point in a plot between recall and efficiency for a given data size, where APC and PSI+X have a point with full recall and high cost, and PSI has a point with low recall and low cost. Here, we will show that \lap allows a trade-off between recall and efficiency for a given privacy guarantee. The efficiency metric used here is the ratio of the cost(\lap) to the cost(APC).

Figure~\ref{fig:vary_eps} illustrates the case when both Alice and Bob require $(\epsilon,\delta)$-\dprl protection where $\epsilon=\{0.1, 0.4,1.6 \}$ and fixed $\delta=10^{-5}$. In Figure~\ref{fig:vary_delta}, we vary the values of $\delta$ for $\delta\in \{10^{-9}, 10^{-7}, 10^{-5}\}$ with fixed $\epsilon=1.6$. Each data point in the plot corresponds to the average cost(\lap)/cost(APC) and average recall of the default \lap for a given $(\epsilon,\delta)$ and the default data size with $T=1$. The default \lap allows the sort \& prune step as described in \S~\ref{sec:prune} with a list of thresholds that are the 90\%, 80\%, ..., 0\% percentiles of the sorted bin sizes of $\tilde{\block}(D_A)$ and $\tilde{\block}(D_B)$. We report the average recall and cost(\lap)/cost(APC) for each percentile. This gives a trade-off line for each $\epsilon$ and $\delta$ value.

We observe that all the trade-off lines obtain a high recall at very small values of cost(LP)/cost(APC). Even at $\epsilon=0.1$, \lap incurs 100 times smaller cost than APC.  \lap has a slightly larger cost for AB dataset.
In Figure~\ref{fig:vary_eps}, the trade-off lines between recall and efficiency shift rightwards as the privacy parameter $\epsilon$ gets smaller. In other words, the cost is higher for a stronger privacy guarantee in order to output the same recall. Similar observations are found in Figure~\ref{fig:vary_delta}. However, the trade-off lines are more sensitive to $\epsilon$ than $\delta$. The red lines in Figure~\ref{fig:vary_eps} and the red lines in Figure~\ref{fig:vary_delta} correspond to the same privacy setting. As $\delta$ reduces by 10000 times from $10^{-5}$ to $10^{-9}$, the trade-off line of \lap for the Taxi datasets shifts the ratio of costs by at most 0.001 as shown in Figure~\ref{fig:vary_delta} (left) while the trade-off line increases the ratio of costs to 0.07 as $\epsilon$ reduces from $1.6$ to $0.1$ (Figure~\ref{fig:vary_eps}).

As the Taxi and AB dataset have different data distributions over bins, the shapes of the trade-off lines are different. AB datasets are more skewed and have some bins with large counts. These bins also have many matching pairs, and hence we see a steep rise for the first part of the trade-off lines for the AB datasets. When the data size increases, if the distribution of matching pairs remains similar, the trade-off lines between the efficiency and recall tends to stay the same. These trade-off lines can be useful when choosing the recall, privacy and efficiency  for larger datasets.

\subsubsection{Wall clock times}\label{sec:clock}
We implemented APC and \lap in python, and implemented operations on encrypted records using the Paillier homomorphic cryptosystem using the python-paillier library \cite{pyphe}. As all algorithms require a one-time encryption of records we exclude this cost and only measure the cost of operations on the encrypted records. On a 3.1 GHz Intel Core i7 machine with 16 GB RAM, we found that computing the Hamming  distance of two encrypted records with dimension $d=50$ takes an average of $t_s = 77$ ms. That is, for datasets of size $n=5000$, APC would take over 22 days to complete! Additionally, for the same dataset with $\epsilon=1.6$, \lap would only take 80 hours to achieve a recall of 1. In comparison, the wall clock time of \lap ignoring the time spent in comparisons of encrypted records was only $120$ seconds. We believe that this order of magnitude difference in time for secure operations and normal operations is true independent of the library or protocol used for secure comparisons.
Thus, the computational cost of \lap is dominated by the cost of secure comparison. How to improve the unit cost of each secure pairwise comparison is an important research topic, and is orthogonal to our research. Hence, in this evaluation, we focused only on the number of secure comparisons/operations on encrypted values to measure efficiency.

\section{Related Work}\label{sec:relatedwork}
In addition to the prior work \cite{Inan:2010:PRM:1739041.1739059, Kuzu:2013:EPR:2452376.2452398,DBLP:conf/icde/CaoRBK15} that attempted to combine DP and secure computation techniques in order to scale-up the \prl problem, there are other efforts that take similar approaches, but focus on solving different problems. Wagh et al. \cite{DBLP:journals/corr/WaghCM16} formalized the notion of \emph{differentially private oblivious RAM (DP ORAM)} and their corresponding protocols significantly improved the bandwidth overheads with a relaxed privacy guarantee. This privacy notion considers a client-server model where all data sit on a single server, while \dprl considers two party computation. Moreover, the protocols for DP ORAM only consider the trade-off between privacy and efficiency while \dprl considers an additional trade-off dimension: correctness. Several efforts \cite{6517175, NIPS2010_0408, Alhadidi:2012:SDF:2359015.2359025, Pettai:2015:CDP:2818000.2818027, Narayan:2012:DDP:2387880.2387895, Goryczka:2013:SMA:2457317.2457343} also integrated DP with SMC in a distributed setting where data is vertically or horizontally partitioned between parties. The difference is that these papers focus on aggregate functions over the partitioned data, such as join size, marginal counts and sum, while \prl requires matching individual record pairs. This matching of individual record pairs does not naturally compose with DP, and hence motivated \dprl, a new privacy model for efficient \prl.

\section{Conclusion} \label{sec:discussion}
In this work, we propose a novel privacy model, called {\em output constrained differential privacy}, that shares the strong privacy protection of differential privacy, but allows for the truthful release of the output of a certain function on the data. We showed that this new privacy model can be applied to record linkage to define differential privacy for record linkage (\dprl).
Under this framework, we proposed novel protocols for efficient \prl that satisfy three desiderata: correctness, privacy and efficiency.  This is an important advance, since none of the prior techniques achieves all three desiderata.
Despite this advance,
further investigation into the practicality of \dprl protocols
is a direction for future research.
This includes investigation into their
wall clock times in a specific operational environment and over datasets with more complex matching functions.
Additional
directions for future research include identifying \dprl protocols that further reduce the computational complexity of record linkage, such as applying a data-dependent blocking strategy,
extending two-party \dprl to a multi-party setting, and generalizing the notion of output constrained differential privacy to other applications beyond private record linkage.

\ \\{\bf Acknowledgements:}
This work was supported by NSF grant 1253327, 1408982, 1443014,
and DARPA \& SPAWAR under contract N66001-15-C-4067.
The U.S. Government is authorized to reproduce and distribute reprints for Governmental
purposes not withstanding any copyright notation thereon.
The views, opinions, and/or findings expressed are those of
the author(s) and should not be interpreted as representing
the official views or policies of the Department of Defense or
the U.S. Government.

\bibliographystyle{ACM-Reference-Format}

\appendix

\section{Related Privacy Definitions} \label{sec:otherprivacy}
\subsection{Simulation-Based S2PC} \label{app:sims2pc}
The standard simulation-based definition for SMC is defined below.
\begin{definition}[SIM-S2PC] \label{def:sim-s2pc}  \cite{Goldreich:2004:FCV:975541}
For a functionality $f$,  a 2-party protocol $\p$ which computes $f$ provides simulation-based secure 2-party computation (SIM-S2PC) if for all data sets $D_A$, $D_B$ of  polynomial sizes (in $\kappa$),
there exist probabilistic polynomial-time algorithms (simulators), denoted by $S_A$ and $S_B$ such that the distribution of $S_A$ (resp., $S_B$) is computationally indistinguishable from $\vw^{\p}_A$ (resp., $\vw^{\p}_B$), i.e.  for any probabilistic polynomial-time (in $\kappa$) adversary $T$,
\begin{eqnarray}
&&Pr[T(S_A(D_A, f_A(D_A,D_B), f(D_A,D_B)))=1]  \\
&\leq & Pr[T(\vw^{\p}_A(D_A,D_B), \out^{\p}(D_A,D_B))=1] + \negligible(\kappa) \nonumber\\
&&Pr[T(S_B(D_A, f_B(D_A,D_B),  f(D_A,D_B)))=1]  \\
&\leq & Pr[T(\vw^{\p}_B(D_A,D_B), \out^{\p}(D_A,D_B))=1] + \negligible(\kappa).  \nonumber
\end{eqnarray}
\end{definition}
If $f$ is deterministic, Alice gains no additional knowledge other than its respective input ($D_A$) and output ($f_A(D_A,D_B)$); similarly for Bob. When randomized functionalities are concerned, augmenting the view of the semi-honest party by the output of the other party is essential. In this case, for any protocol $\p$ that computes the randomized functionality $f$, it does not necessarily hold that $\out^{\p}(D_A,D_B) = f(D_A,D_B)$. Rather, these two random variables must be identically distributed.  In order to study the possibility of composing DP and S2PC, we choose the indistinguishability-based definition for PRL, which is implied from SIM-S2PC.
\begin{theorem} \label{theorem:sim-ind}
SIM-S2PC implies IND-S2PC.
\end{theorem}
\begin{proof}
Given the protocol $\p$, for all possible inputs $(D_A, D_B)$,  there exists a global simulator $S_A$ such that the distribution of $S_A$ is computationally indistinguishable from the view of Alice. As $f(D_A,D_B) = f(D_A,D'_B)$, $S_A$ takes the same input and hence will have the same distribution for $D_B$ and $D'_B$. Hence, the views over $(D_A,D_B)$ or $(D_A,D'_B)$ are indistinguishable.
\end{proof}

Any algorithm that satisfies SIM-CDP also satisfies IND-CDP \cite{Mironov:2009:CDP:1615970.1615981}, but it is unknown if the converse holds.

\subsection{Computationally Differential Privacy}\label{app:cdp}
Mironov et al. \cite{Mironov:2009:CDP:1615970.1615981} defines a privacy notion, known as indistinguishable computationally differential privacy (IND-CDP-2PC). This notion is a direct extension of DP in two party setting where both parties are computationally bounded. Formally, we have
\begin{definition}[IND-CDP-2PC] \label{def:ind-cdp-2pc}
A 2-party protocol $\p$ for computing function $f$ satisfies $(\epsilon_A(\kappa),\epsilon_B(\kappa))$-indistinguishable computationally differential privacy (IND-CDP-2PC) if $ \vw^{\p}_A(D_A,\cdot)$ satisfies $\epsilon_B(\kappa)$-IND-CDP, i.e. for any probabilistic polynomial-time (in $\kappa$) adversary $T$, for any neighboring databases $(D_B,D'_B)$ differing in a single row,
\begin{eqnarray}
&&Pr[T(\vw^{\p}_A(D_A,D_B))=1] \nonumber \\
&\leq & e^{\epsilon_B} Pr[T(\vw^{\p}_A(D_A,D'_B))=1] + \negligible(\kappa).
\end{eqnarray}
The same holds for Bob's view for any neighbors $(D_A,D'_A)$ and  $\epsilon_A$.
\end{definition}

\subsection{Local Differential Privacy}\label{app:localdpdef}

The local model is usually considered in the model where individuals do not trust the curator with their data. The local version of differential privacy is defined as follows.

\begin{definition}[$\epsilon$-Local Differential Privacy] \cite{Dwork:2014:AFD:2693052.2693053}
A randomized mechanism $M:\dom \rightarrow \mathcal{O}$ satisfies $\epsilon$-local differential privacy if
\begin{equation}
\Pr[M(r)=O] \leq e^{\epsilon} \Pr[M(r')=O]
\end{equation}
for any set $O\subseteq \mathcal{O}$, and any records $r,r'\in \dom$ and $\epsilon>0$.
\end{definition}

\section{Theorems \& Proofs}
\subsection{Privacy Leakage in Prior Work}
\subsubsection{Theorem~\ref{theorem:lsh} (Limitations of \prl with Blocking)} \label{app:prlblocking}
Given $(d_1,d_2,p_1,p_2)$-sensitive $H=\{h_0,..,h_{|H|-1}\}$, we use $H(\cdot)$ for a record to denote the list of hashing values $[h_0(\cdot),\ldots, h_{|H|-1}(\cdot)]$. An LSH-based blocking considers a set of bins where records associated with the same value for $H(\cdot)$ are hashed to the same bin. A popular blocking strategy is to compare all the corresponding bins, and results in a set of candidate matches $\{(a,b)| h(a)=h(b) \forall h\in H, a\in D_A, b\in D_B \}$, i.e.
$\{(a,b)| H(a)=H(b), \forall a\in D_A, b\in D_B \}$. We can show that any LSH based blocking cannot satisfy IND-S2PC as stated in Theorem~\ref{theorem:lsh}. Here is the proof.
\begin{proof}
Take a pair of  databases $(D_B,D'_B)$ where $f_{\Join_m}(D_A,D_B) = f_{\Join_m}(D_A,D'_B)$. Let the symmetric difference between $D_B$ and $D'_B$ be $(b,b')$ and $dist(b,b')>d_2$. Hence, with high probability $1-p_2^{|H|}$, we have $H(b)\neq H(b')$, and
$|\block_{H(b)}(D_B)| - |\block_{H(b)}(D'_B)|=1$  and $|\block_{H(b')}(D'_B)| - |\block_{H(b')}(D_B)|=1$ as the rest of records are the same in $D_B$ and $D'_B$. Alice as a semi-honest adversary can set her dataset such that $ |\block_{H(b)}(D_A)| \neq |\block_{H(b')}(D_A)|$. Then, with high probability, the following inequality holds
\begin{eqnarray}
&&cost_{\block^S}(D_A,D_B) - cost_{\block^S}(D_A,D'_B) \nonumber \\
&= & (|\block_{H(b)}(D_B)| - |\block_{H(b)}(D'_B)| )| \block_{H(b)}(D_A)|  \nonumber \\
&& + (|\block_{H(b')}(D_B)| - |\block_{H(b')}(D'_B)| )|\block_{H(b')}(D_A)| \nonumber \\
&=& |\block_{H(b)}(D_A)| - |\block_{H(b')}(D_A)| \neq 0.
\end{eqnarray}
Hence, Alice can distinguish $D_B$ and $D'_B$ by $cost_{\block^S}(D_A,D_B) \neq cost_{\block^S}(D_A,D'_B)$ with high probability $1-p_2^{|H|}$. Other blocking strategies can be similarly shown. Therefore, this LSH-based PRL does not satisfy IND-S2PC.
\end{proof}

\subsubsection{Theorem~\ref{theorem:prldp_limit} (Limitations of \prl with DP Blocking of Prior Approaches/\lap-2)}\label{app:prldp}
Several prior works \cite{Inan:2010:PRM:1739041.1739059, Kuzu:2013:EPR:2452376.2452398,DBLP:conf/icde/CaoRBK15} combine \prl techniques with differentially private blocking (\prl+$\block_{DP}$). These approaches can be summarized in three steps: (1) DP blocking, (2) records addition and suppression, (3) secure pair-wise comparisons based on blocking strategy $\block^{S}$. In the first step, Alice and Bob process their data independently. Each party generates an $\epsilon$-differentially private partition of the data, where each partition is associated with a noisy count $\tilde{o}_i = |\block_i(D_B)|+\eta_i$, where $\Pr[\eta_i=x] = pe^{-\epsilon/\Delta\block \cdot |x|}$, for $x\in \mathbb{Z}$ and $p=\frac{e^{\epsilon/\Delta\block-1}}{e^{\epsilon/\Delta\block}+1}$ is the normalized factor  \footnote{We use discrete version of Laplace distribution to avoid rounding.}. $\Delta\block$ is the sensitivity of the blocking strategy (Def~\ref{def:deltablock}).

Next, for each partition, if the noise $\eta_i$ is positive, dummy records are added; otherwise, records in that partition are suppressed randomly to obtain the published count. This results in new bins, denoted by $\{ \tilde{\block}_i(D_A)\}$ and $\{ \tilde{\block}_j(D_B)\}$. In the last step, Alice and Bob jointly compare record pairs $(a,b)$, where $a\in \tilde{\block}_i(D_A)$ and $b\in \tilde{\block}_j(D_B)$ for all $(i,j)\in \block^S$ as in APC. They only exchange the true records $(a,b)$ if they match. \cite{DBLP:conf/icde/CaoRBK15} considers a third party for identifying candidate pairs for Alice and Bob, so that Alice and Bob has no direct access to the noisy bins of the opposite party, but has access to the number of secure comparisons. However, this hybrid protocol above does not satisfy $(\epsilon,\delta)$-DPRL as stated in Theorem~\ref{theorem:prldp_limit}. The failure to satisfy \dprl is mainly caused by the record suppression step for the negative noise drawn from a zero-mean Laplace distribution, as shown in the following proof.

\begin{proof}
Without loss of generality, we consider Alice as the adversary. For any arbitrary $\epsilon$ and small $\delta< \frac{p^{\Delta\block}}{2e^{\epsilon}}$, there exists a counter example fails $(\epsilon,\delta)$-\dprl. For simplicity, we illustrate how to construct counterexamples using a blocking strategy $\block$ with sensitivity $\Delta\block=2$, where Alice and Bob use the same hashing and each record is hashed to at most 1 bin. For other blocking strategies, counterexamples can be similarly constructed.

Fix a $D_A$, consider $D_B$ such that $\block_0(D_B) = \{b_*\}$ and $\block_1(D_B) = \{b_1,..,b_{n_1}\}$, where $1\leq n_1 < \frac{p^2}{e^{\epsilon}\delta-1}$. (Note that $\frac{p^2}{e^{\epsilon}\delta}>2$ because $\delta< \frac{p^2}{2e^{\epsilon}}$.) In addition, all records in $\block_1(D_B)$ can find some matching ones from $D_A$, but $b_*$ does not match any record in $D_A$. A neighboring database $D'_B$ can be constructed from $D_B$ by removing $b_*$ from $\block_0$, and adding another  $b'_*$ that can be hashed to $\block_1$.  It is easy to see that $(D_B,D'_B)\in \mathcal{N}(f_{\Join_m}(D_A,\cdot))$.

Without a third party \cite{Inan:2010:PRM:1739041.1739059, Kuzu:2013:EPR:2452376.2452398}, Alice and Bob has access to the number of secure comparisons and the noisy bin counts (in addition to the input data sizes and the matching output). Consider a set of views of Alice $\vw^*$ with output that contains all matching pairs from $\block_1(D_A)\Join_m \block_1(D_B)$ and noisy counts for bin $\block_0$ and $\block_1$ for Bob being $0$ and $n_1$ respectively. Let $\tilde{\block}(D_A)$ be the noisy bins that Alice uses for the final secure pairwise comparisons.  The probabilities to generate these views from $D_B$ and $D'_B$ are respectively:
\begin{eqnarray}
\Pr[\vw^* |\tilde{\block}(D_A), D_B] &=&  \Pr[\eta_0=0] \Pr[\eta_1=0]= p^2, \\
\Pr[\vw^* |\tilde{\block}(D_A), D'_B] &=& \Pr[\eta_0=1]\Pr[\eta_1=-1 \& \textit{~suppress~}b'_*]   \nonumber \\
&=& p^2 / (e^{\epsilon} (n_1+1)) > \delta
\end{eqnarray}
The inequality above is due to $n_1<\frac{p^2}{e^{\epsilon}\delta}-1$. Hence, we have
\begin{eqnarray}
\Pr[\vw^* |\tilde{\block}(D_A), D_B] &=& (e^{\epsilon} +n_1 e^{\epsilon}) \Pr[\vw^* |\tilde{\block}(D_A), D'_B] \nonumber \\
&>& e^{\epsilon}\Pr[\vw^* |\tilde{\block}(D_A), D'_B] + \delta.
\end{eqnarray}
Hence, $(\epsilon,\delta)$-DPRL is violated.

With a third party\cite{DBLP:conf/icde/CaoRBK15}, Alice and Bob has access to the final output, and the total number of secure pairwise comparisons,
but not the noisy bin counts. We can construct examples where knowing the number of secure comparisons leaks the noisy bin counts. After which the previous arguments (for the case with no third party) can show that this protocol does not ensure DPRL for all epsilon and delta.  For instance, consider Alice has only 1 record in $\tilde{\block}_0(D_A)$, and more than 1 records in other bins, if the output $O= D_A\Join_m D_B$, and the total number of secure pairwise comparisons is $|O|+1$. This secure pairwise matching that returns false can  only happen between a record of Bob from $\tilde{\block}_0(D_B)$ with the record from $\tilde{\block}_0(D_A)$. Hence, Alice can infer the noisy counts of $\tilde{\block}(D_B)$. Then the argument for the case with no third party can be used.
\end{proof}

In addition, by Theorem~\ref{theorem:dprl_ind-s2pc} and Theorem~\ref{theorem:dprl_cdp}, \dprl is weaker than IND-S2PC and IND-CDP-2PC, we have the following result.

\begin{corollary}\label{cor:lap2}
\lap-2 satisfies neither IND-CDP-2PC nor IND-S2PC.
\end{corollary}

\subsection{Properties of Output Constrained DP}
\subsubsection{Theorem~\ref{theorem:seq} (Sequential Composition)}\label{app:seqproof}
\begin{proof}
Consider Alice as a probabilistic polynomial-time (in $\kappa$) adversary $T$,
with input $D_A$. $(D_B,D'_B)$ are neighbors w.r.t. $f(D_A,\cdot)$. We have the
probabilities of distinguishing $D_B$ and $D'_B$ bounded by
\begin{eqnarray} \label{eq:p1}
&&\Pr[T(\vw^{\p_2,\p_1}_A(D_A,D_B))=1] \nonumber \\
&\leq& \int_x \Pr[T(\vw^{\p_2}_A(D_A,D_B,x))=1]\cdot \nonumber\\
&&      \Pr[x=\vw^{\p_1}_A(D_A,D_B)]dx \nonumber\\
&\leq& \int_x (e^{\epsilon_2}\Pr[T(\vw^{\p_2}_A(D_A,D_B',x))=1] + \delta_2)\cdot \nonumber \\
&&     \Pr[x=\vw^{\p_1}_A(D_A,D_B)]dx \nonumber\\
&\leq& \int_x (e^{\epsilon_2}\Pr[T(\vw^{\p_2}_A(D_A,D_B',x))=1]) \cdot\nonumber \\
&&     (e^{\epsilon_1}Pr[x=\vw^{\p_1}_A(D_A,D'_B)]+\delta_1)dx +\delta_2 \nonumber\\
&\leq& e^{\epsilon_1+\epsilon_2} \Pr[T(\vw^{\p_2,\p_1}_A(D_A,D'_B))=1] + \delta_1+\delta_2\nonumber
\end{eqnarray}
\end{proof}

\subsubsection{Theorem~\ref{theorem:pp} (Post-processing)}\label{app:ppproof}
\begin{proof}
Since $g$ is efficient and in composition with $T$ can be used as adversary itself. If $g(\out^{\p}(D_A,D_B))$ does not satisfy $(\epsilon,\delta)$-IND-DPRL, then $\p$ does not satisfy $(\epsilon,\delta)$-IND-DPRL.
\end{proof}

\subsubsection{Theorem~\ref{theorem:dprl_cdp} (Relation with IND-CDP-2PC)}
We show that \dprl is weaker than IND-CDP-2PC.
\begin{proof}
$\epsilon/2$-IND-CDP-2PC is equivalent to $\epsilon$-IND-DP-2PC, where neighboring databases have a symmetric difference of 2.
The set of neighboring databases for \dprl is a subset of that for $\epsilon$-IND-DP-2PC, and hence $(\epsilon,\delta)$-\dprl is weaker than $\epsilon/2$-IND-DP-2PC.
\end{proof}

\subsection{Properties for \dprl Protocols}

\subsubsection{Theorem~\ref{theorem:lap_dprl} (Privacy of Laplace Protocol)}
\begin{lemma} \label{lemma:ratio}
With probability $1-\delta$, the probability  for Alice having the same view from
neighboring databases $(D_B,D'_B)\in \mathcal{N}(f_{\Join_m}(D_A,\cdot))$ is bounded by $e^{\epsilon}$.
\end{lemma}
\begin{proof}
Given $(D_B,D'_B)\in \mathcal{N}(f_{\Join_m}(D_A,\cdot))$ and $\block$, the maximum difference in the bin counts of $D_B$ and $D'_B$ is $\Delta\block$. Let $\block_{\Delta}$ be the set of bins that $D_B$ and $D'_B$ have different counts, and $\sum_{i\in \block_{\Delta}} |\block_i(D_B)-\block_i(D'_B)| \leq \Delta\block$.
If all the noise for these bins are non-negative, then the probability to output the same noisy counts $(c_0,\ldots,c_{k-1})$ from $D_B$ and $D'_B$ is bounded by
\begin{eqnarray}
&&\ln(\frac{\Pr[(c_0,\ldots,c_{k-1}|D_B]}{\Pr[(c_0,\ldots,c_{k-1}|D'_B]}) \nonumber \\
&=&\ln (\frac{\prod_{i=0}^{k-1}\Pr[\eta_i=c_i-|\block_i(D_B)|]}{\prod_{i=0}^{k-1}\Pr[\eta_i=c_i-|\block_i(D'_B)|]}) \nonumber \\
&=&\sum_{i\in \block_{\Delta}} \ln(\Pr[\eta_i = c_i-|\block_i(D_B)|])-\ln(\Pr[\eta_i = c_i-|\block_i(D'_B)|]) \nonumber \\
&\leq & \epsilon/\Delta\block \cdot (\sum_{i\in \block_{\Delta}}|\block_i(D_B)-\block_i(D'_B)|) \leq \epsilon  \nonumber
\end{eqnarray}
The probability to draw a negative noise $\eta$ from $Lap(\epsilon,\delta,\Delta\block)$ is
\begin{eqnarray}
Pr[\eta < 0] = \sum_{i=-1}^{-\infty} p \cdot e^{-(\epsilon/\Delta\block) (x-\eta^0)} =
 \frac{e^{-\eta_0\epsilon/\Delta\block}}{e^{\epsilon/\Delta\block} +1}
\end{eqnarray}
Given $\eta_0=-\frac{\Delta\block\ln ((e^{\epsilon/\Delta\block}+1)(1-(1-\delta)^{1/\Delta\block}))}{\epsilon}$, we have $Pr[\eta < 0] = 1- (1-\delta)^{1/\Delta\block}$.
For each neighboring pair, at most $\Delta\block$ bins differ and fail to have $Pr[\eta\geq 0]$. Hence, the overall failing probability is $1-(1-Pr[\eta < 0])^{\Delta\block}=\delta$. With $1-\delta$, the probability of having the same view from
\prl neighboring databases is bounded by $e^{\epsilon}$.
\end{proof}

\subsubsection{Theorem~\ref{theorem:gmc_eff_recall} (Correctness \& Efficiency of GMC)}
\label{app:gmc_eff_recall}
\begin{proof}
First, we will show that the efficiency of \lap with the greedy match \& clean step (GMC) is better than \lap alone. The first part of the protocol that adds dummy records is the same. The second part of the protocol without GMC compares all the candidate matches using the secure matching protocol $SMC(a,b)$. On the other hand, with GMC, if a record pair $(a,b)$ is compared securely, then $(a,b)$ must be one of the candidate matches. Hence, the number of the secure pairwise comparisons with GMC will be no more than the protocol without GMC.

Next, we will show the correctness of \lap with GMC. Let $O_{\lap}, O_{\lap+GMC}$ be the final output of \lap protocol without GMC and with GMC. We would like to show that if $(a,b)\in O_{\lap}$, then $(a,b)\in O_{\lap,GMC}$. Suppose this is not true, then there exists a matching pair $(a,b)\in O_{\lap}$, but  $(a,b)\notin O_{\lap+GMC}$. If so, then one of the records in $(a,b)$ must be removed from the bins before its turn of secure pairwise comparison $SMC(a,b)$. Without loss of generality, let's say $a$ is cleaned from Alice's bins before $SMC(a,b)$. The condition to remove $a$ is that $a$ has already been in the current output. Hence, Bob is able to compare $a$ with all his records in plain text and identify this matching pair $(a,b)$. This leads to a contradiction. Hence, $O_{\lap}\subseteq O_{\lap+GMC}$. Moreover, if a matching pair $(a,b)$ is not a candidate match based on the blocking strategy $\block^S$, and if $a$ has been already found matching with another record of Bob, then GMC can add $(a,b)$ into $O_{\lap+GMC}$. Hence, it is possible that \lap with GMC gains even more matching pairs than \lap alone.
\end{proof}

\section{Additional Protocols}
\subsection{Example for secure pairwise match} \label{app:eg_smc}
Here we give an example for the function $SMC(a,b)$ that outputs $(a,b)$ if they match; null otherwise. The matching rule is that Euclidean distance of $a$ and $b$ is less than $\theta$.
First, Party Alice creates a homomorphic public/private key pair $(pk,pr)$, and sends the public key $pk$ to party Bob. Let $E_{pk}(\cdot)$ denote the encryption function with public key $pk$ and $D_{pr}(\cdot)$ the decryption function with private key $pr$. Paillier's cryptosystem supports the following operations on the encrypted plain texts $m_1$ and $m_2$ without the knowledge of the private key:
\begin{itemize}
 \item Addition: $E_{pk}(m_1+m_2) = E_{pk}(m_1) +_h E_{pk}(m_2)$;
 \item Multiplication with constant $c$: $E_{pk}(cm_1) = c \times_h E_{m_1}$
\end{itemize}
These two operations allow secure computation of Euclidean distances, i.e. $dist(a, b) = \sum_{i}(a[i]-b[i])^2 = \sum_{i}(a[i])^2-2a[i]b[i]+(b[i])^2$, and also hamming distances for bit vectors, i.e. $dist(a,b)=  \sum_{i} |a[i]-b[i]|=  \sum_{i} a[i]+b[i]-2a[i]b[i]$.

As summarized in Algorithm~\ref{algo:securematch}, given bit vectors $a$ and $b$, Alice will send to Bob the encrypted values $(a_{id}, \left\{ E_{pk}(a[i]) \right\})$ where $a_{id}$ is a randomly generated record identifier for record $a$. Next, party Bob computes for each of its records $b_{id}$ the value $E_{pk}(a[i]) +_{h} E_{pk}(a[i]) \times_{h} (-2b[i])+_{h} E_{pk}(b[i])$ which is equal to $E_{pk}(|a[i]-b[i]|)$ for all $i$, and computes the encrypted $E_{pk}(\sum_i |a[i]-b[i]|)$. A random number $r$ is generated and added to the encrypted distance, such that the true distance is hidden from Alice if $(a,b)$ is not a matching pair. Party Bob creates the message $(b_{id}, E_{pk}(\sum_i |a[i]-b[i]| + r)$ for each record pair comparison. Alice can then decrypt the message with her private key and obtain the relative distance $d=\sum_i |a[i]-b[i]| + r$. Since Bob knows $\theta+r$, a secure comparison protocol, such as Yao's garbled circuit \cite{Yao:1986:GES:1382439.1382944},  can be used to evaluate if $d\leq \theta+r$. If this algorithm outputs ``True'', Alice and Bob will exchange their true record values.
\begin{algorithm}[t]
    \SetKwInOut{Input}{Input}
    \SetKwInOut{Output}{Output}
    \underline{function SecureMatch} $(a,b, \theta)_{pk,pr}$\;
    \Input{$ a,b \in \left\{0,1\right\}^d$, hamming distance threshold $\theta$, public/private key pair $(pk,pr)$}
    \Output{$(a,b)$ or $\emptyset$}
    \ Alice: randomly generates an id $a_{id}$ and sends to Bob\;
    \ Bob: randomly generates an id $b_{id}$ and an integer $r$\;
    \ Bob: initiates $s$ with $E_{pk}(r)$\;
    \For{$i \in [0,\ldots,d-1]$}{
    		\ Alice: sends to Bob $E_{pk}(a[i])$\;
    		\ Bob: updates $s = s +_h E_{pk}(a[i]) +_{h} (E_{pk}(a[i]) \times_{h} (-2b[i]))+_{h} E_{pk}(b[i])$\;
    }
    \ Bob: sends $(b_{id},s)$ to Alice \;
    \ Alice: decrypts $s = D_{pr}(s)$  \;
    \eIf{$s \leq \theta+r$ (secure integer comparison)}
      {
        return $(a,b)$\;
      }
      {
        return $\emptyset$\;
      }
    \caption{Secure Match $a$ and $b$}\label{algo:securematch}
\end{algorithm}

\subsection{Local DP Protocol}\label{app:localdp}
\subsubsection{Theorem~\ref{theorem:localdp}}
Let $\block$ be a blocking that randomly hashes records into a pre-specified set of $k$ bin, such that $\frac{\Pr[\block(b) = i]}{\Pr[\block(b')=i]} \leq e^{\epsilon}$.
Such a blocking $\block$ satisfies $\epsilon$-local DP (Appendix~\ref{app:localdpdef}). Protocols that combine a local differentially private blocking with IND-S2PC protocols for record linkage to achieve $(\epsilon, \negligible(\kappa))$-\dprl.
\begin{proof}(sketch) We prove privacy for Bob (the proof for Alice is analogous). In this protocol, Alice with input Data $D_A$ has a view consisting of (1) the number of candidate matching pairs arising in each $(i,j)\in \block^S$, (2) the output for each candidate matching pair. Consider a neighboring pair $(D_B,D'_B) \in N(f_{\Join_m}(D_A,\cdot))$ for a given $D_A$. By Theorem~\ref{theorem:neighbor_prl}, $D_B$ and $D'_B$ differ in only one non-matching record with respect to $D_A$, i.e. $D'_B=D_B-b_*+b_*'$ and $b_*\neq b_*'$, where $m(b_*,a)=0$ and $m(b_*',a)=0$ for all $a \in D_A$. Given both $b_*,b_*'$ can be hashed into the same bin with probability ratio bounded by $e^\epsilon$, the probabilities of generating the same number of candidate matching pairs from $D_B$ and $D'_B$  are also bounded by the same ratio. The encrypted records only differ in $b_*$ and $b'_*$, and both of them lead to the same output for each candidate matching, because they do not match any records in $D_A$.
\end{proof}

In this work, we use randomized response  (RR) \cite{Dwork:2014:AFD:2693052.2693053} as an example to achieve \dprl. Other local DP algorithms, such as Johnson-Lindenstrauss (JL) transform \cite{Blocki:2012:JTI:2417500.2417899} can be similarly applied.

\subsubsection{RR based Blocking} \label{app:rr_basic}

Given a fixed hash function $h: \dom \rightarrow [0,k-1]$, records in $D_B$ are hashed into $k$ bins, $\block_0,..,\block_{k-1}$ respectively. Let us define a RR based on this fixed hashing function with privacy budget $\epsilon_B$ for Bob. Each record $b\in D_B$ is randomly hashed into $\block_{h(b)}$ with probability $p_B = \frac{e^{\epsilon_B}}{k-1+e^{\epsilon_B}}$ and the other $(k-1)$ bins with probability $q_B = \frac{1}{k-1+e^{\epsilon_B}}$. We denote the resulted bins by $\tilde{\block}(D_B)$, and the resulted bin for each record $b\in D_B$ by $\tilde{\block}(b)$.
Similarly, using the same fixed hashing function and randomized response, Alice's records $D_A$ are randomly hashed into the $k$ bins, $\tilde{\block}(D_A)$ with corresponding $p_A,q_A$ based on $\epsilon_A$. This randomized response with probabilities $(p_A,q_A)$ and $(p_B,q_B)$ ensures $\epsilon_A$-local DP and $\epsilon_B$-local DP respectively.

Consider a basic blocking strategy $\block^S = \{(i,i)| i\in [0,k)\}$, all corresponding bins are compared. If the hash function $h$ is a LSH, then matching records are likely fall into the same bin as Alice and Bob use the same hash function. The probability that such records $(a,b)$ appear in the same bin after randomization $\Pr[\tilde{\block}(a)=\tilde{\block}(b) | h(a)=h(b)]$ is $p_Ap_B + (k-1)q_Aq_B$. This probability  increases with the privacy budget $\epsilon_A,\epsilon_B$, and hence recall will improves.

A further trade-off between correctness, privacy and efficiency is allowed by considering a general blocking strategy $\block^S = \{(i,(i+j)\%k)| i\in [0,k), j\in [0,k') \}$,  each record in $\tilde{\block}_i(D_A)$ is securely compared with $k'$ neighboring bins of Bob. The basic blocking strategy corresponds to the case where $k'=1$. As $k'$ increases, more candidate matching pairs are securely compared, resulting potentially higher recall and more communication and computation cost. When $k'=k$, the resulted protocol is equivalent to the all pairwise comparisons baseline.

\subsubsection{Optimal RR Probability}

If the window size $k'$ is given in advance as a parameter for the efficiency, the expected recall of randomized response can be further optimized.
Let $p^B_i$ be the probability for a record $b \in D_B$ to be randomly hashed into $\block_{(h(b)+i)\% k}$. To ensure $\epsilon_B$-DPRL, the probabilities to hash any pairs of bins should be bounded by $e^{\epsilon_B}$. The hashing probability for Alice records, $p^A_i$, is similarly defined and constrained. As each record of Alice's can fall into any bin, and each bin of Alice's is compared with $k'$ neighboring bins of Bob's, the probability that a matching pair $(a,b)$ is compared after randomization is
\begin{equation} \label{eqn:pab}
p_{(a,b)}=\sum_{i=0}^{k-1}\sum_{j=0}^{k'-1}p_i^Ap_{i+j\%k}^B.
\end{equation}
The expected recall can be improved by maximizing $p_{(a,b)}$ (Eqn.~\eqref{eqn:pab}) with the constraints on
\begin{itemize}
\item[(a)] Ratio: $\frac{p^B_i}{p^B_{i'}}\leq e^{\epsilon_B}, \forall i\neq i'$; $\frac{p^A_i}{p^A_{i'}}\leq e^{\epsilon_A}, \forall i\neq i'$;
\item[(b)] Sum: $\sum_{i=0}^{k-1} p_i^A=1$; $\sum_{i=0}^{k-1}p_i^B=1$.
\end{itemize}

The solution in Section~\ref{app:rr_basic} where we assign a high probability to a single bin, and a low probability to the rest of the bins is a valid solution to the maximization problem, but it is not always the optimal solution. In general, this optimization can be solved by existing tools for quadratic objectives with linear constraints, such as quadratic programming.

Here, we present an explicit solution for the special case when $\epsilon_B=\epsilon \geq 0$ and $\epsilon_A\rightarrow \infty$. This case corresponds to the situation where Alice's data is public and only Bob's data requires $\epsilon$-DPRL protection.
Then the expected recall can be maximized by solving the following linear optimization:
\[ \max_{p^B_0,\ldots,p^B_{k-1}}   \sum_{i=0}^{k'-1} p^B_i ~~~\text{s.t. }
 \sum_{i=0}^{k-1} p^B_i = 1 \text{, and }  \frac{p^B_{i}}{p^B_{i'}} \leq e^{\epsilon} ~~ \forall i\neq i'
\]
The expected recall is maximized with value $p_{(a,b)}=\frac{k'e^{\epsilon}}{k-k'+k'e^{\epsilon}}$ at
\[
    p^B_i=
\begin{cases}
    \frac{e^{\epsilon}}{k-k'+k'e^{\epsilon}},& \text{for }i=0,\ldots,k'-1\\
    \frac{1}{k-k'+k'e^{\epsilon}},              & \text{for }i=k',\ldots,k-1
\end{cases}
\]

If Bob's records are uniformly distributed over the bins, then the compression ratio w.r.t all pairwise comparisons (APC) is $\rho= \frac{k'}{k}$. The maximized expected recall can be written as
\begin{equation}\label{eqn:irr}
p_{(a,b)} = \frac{\rho e^{\epsilon}}{1-\rho+\rho e^{\epsilon}}.
\end{equation}
This equation explicitly form the relationship between correctness ($p_{(a,b)}$), privacy ($\epsilon$) and efficiency ($\rho$) of this protocol.

\begin{theorem}\label{theorem:localdp_constantfactor}
The basic RR mechanism achieves a constant factor speedup in efficiency given $\epsilon_B=\epsilon \geq 0$, $\epsilon_A\rightarrow \infty$, $\delta_A=\delta_B = \negligible(\kappa)$ and recall $r$.
\end{theorem}
\begin{proof}
Given a recall $r=p_{(a,b)}$, we have $\rho= 1 - \frac{e^{\epsilon}(1-r)}{r+e^{\epsilon}(1-r)}$ based on Eqn.~\eqref{eqn:irr}.
The improvement in efficiency $\rho$ is a constant factor in terms of $r$ and $\epsilon$, independent of $n$.
\end{proof}

\newcommand{\hi}{\top}
\newcommand{\low}{\bot}

\balance

Next, we compute the optimal amongst a restricted class of strategies for the  more general case where $\epsilon_A=\epsilon_B=\epsilon$, though the explicit form for this case is unknown yet. The strategies we consider are those where Alice and Bob (a) use symmetric probabilities to assign a bin to each record, and (b) they both assign a high probability $p_\hi$ to place a record from bin $i$ to bins $i$ through $(i+x-1)\%k$ (for some $1 \leq x \leq k'$), and a low probability $p_\low$ to assign a record from bin $i$ to the rest of the bins. Note that, when only one of Alice or Bob is randomizing their records, $x=k'$ results in the RR probabilities that optimize the expected recall.

In order to satisfy the constraints in the above maximization problem, we need
$p_\hi = \frac{e^\epsilon}{k-x + x\cdot e^\epsilon}$ and $p_\low = \frac{1}{k-x + x\cdot e^\epsilon}$.

The expected recall can be derived as follows. Without loss of generality consider a matching pair $(a,b)$ that fall into bin 0. There are 3 ways $(a,b)$ are matched after randomization:\\
{\em (i): } Both Alice and Bob randomize their records to a $\hi$ bin (i.e., some bin  $0 \leq j \leq x-1$). Since $x \leq k'$, these records are definitely compared. This occurs with probability $\frac{x(x+1)}{2}p_\hi^2$\\
{\em (ii): } Only one of Alice and Bob randomize their records to a $\hi$ bin (i.e., some bin  $0 \leq j \leq x-1$). There $2k'x - x(x+1)$ ways in which exactly one of $a$ or $b$ is randomized to a $\hi$ bin, but still end up getting compared by the algorithm. This occurs with probability $(2k'x - x(x+1)) p_\hi \cdot p_\low$\\
{\em (ii): } Both Alice and Bob randomize their records to a $\low$ bin. This occurs with probability $(kk' - (2k'x - \frac{x(x+1)}{2}))p_\low^2$
In total, the probability that $(a,b)$ are compared is expressed in terms of $x$ as
\begin{eqnarray}
p_{(a,b)}(x) &=&\frac{x(x+1)}{2}p_\hi^2 + (2k'x - x(x+1)) p_\hi \cdot p_\low \nonumber \\
&& + (kk' - (2k'x - \frac{x(x+1)}{2}))p_\low^2 \nonumber
\end{eqnarray}
The derivative of $p_{(a,b)}(x)$ w.r.t $x$ is
\begin{equation}
p'_{(a,b)}(x) = C_1 \cdot [(e^{\epsilon}-1)C_2x + k(2k'+e^{\epsilon}-1)], \nonumber
\end{equation}
where $C_1=\frac{e^{\epsilon}-1}{2(k-x+xe^{\epsilon})^3}$ and $C_2 = (e^{\epsilon}-3+2k-4k')$.
When $C_2>0$, the derivative is always positive, the expected recall is maximized when $x=k'$, as $0<x\leq k'$. We will leave the complete analysis to the future work.

\section{Additional Plot}\label{app:eval}
Figure~\ref{fig:greedy_recall95} shows the log(base 10) value of the average cost with respect to the $\log$ value of data size for PSI+X, APC, and \lap with $\epsilon\in\{0.1, 0.4, 1.6\}$ and $ \delta=10^{-5}$ and the non-private setting (np) when they achieve a recall $>0.95$. Similar to Figure~\ref{fig:greedy_recall100}, \lap gives lower costs than the baselines, and scales near linearly.

\begin{figure}[H]
\centering
\includegraphics[scale=0.53]{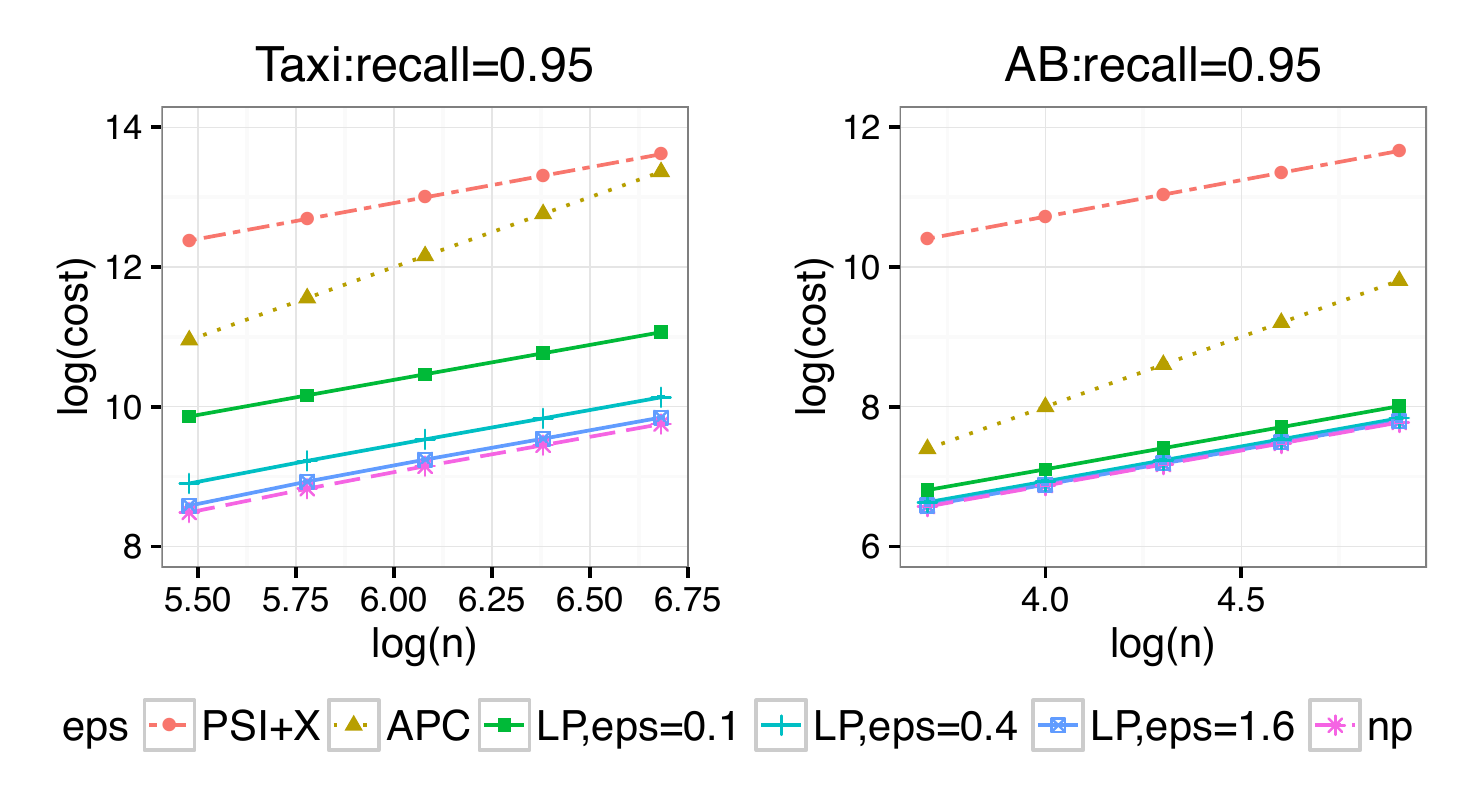}
\caption{The average $\log(cost)$ vs $\log(data~size)$} \label{fig:greedy_recall95}
\end{figure}

\end{document}